\documentclass[]{interact}

\usepackage{epstopdf}% To incorporate .eps illustrations using PDFLaTeX, etc.
\usepackage{subfigure}% Support for small, `sub' figures and tables

\usepackage{graphicx}
\usepackage{amsmath}
\usepackage{amsfonts}
\usepackage{multirow}
\usepackage{amssymb}
\usepackage{color}
\usepackage{hyperref}
\usepackage{longtable}
\usepackage{natbib}
\usepackage{enumitem}
\bibpunct[, ]{(}{)}{;}{a}{}{,}% Citation support using natbib.sty
\renewcommand\bibfont{\fontsize{10}{12}\selectfont}% Bibliography support using natbib.sty

\usepackage{url}
\newtheorem{lemma}{Lemma}
\newtheorem{theorem}{Theorem}
\newtheorem{definition}{Definition}

\newtheorem{remark}{Remark}
\newtheorem{problem}{Problem}
\newtheorem{assumption}{Assumption}
\newtheorem{corollary}{Corollary}

\newcommand{\rr}{\mathbb{R}}
\newcommand{\cc}{\mathbb{C}}

\begin{document}

%\articletype{ARTICLE TEMPLATE}

% \bibliographystyle{unsrt}
%\baselinestretch
% \begin{frontmatter}
%\title{A parameter-dependent LMI approach for cooperative synchronization of agents with delays}
\title{Robust cooperative synchronization of homogeneous agents with delays on directed communication graphs} % -based synchronizing region, A parameter-dependent LMI approach
%\author{Baran Alikoc, Kristian Hengster-Movric and Michael Sebek}
%\ead{Baran.Alikoc@cvut.cz, Kristian.Hengster.Movric@fel.cvut.cz, Michael.Sebek@cvut.cz}
%\address{Czech Institute of Informatics, Robotics and Cybernetics, Czech Technical University in Prague, Czech Republic\\
%Department of Control Engineering, Czech Technical University in Prague, Czech Republic
%}

%\articletype{ARTICLE TEMPLATE}

\author{
\name{Baran Alikoc\textsuperscript{a}\thanks{Corresponding Author: B. Alikoc, e-mail: baran.alikoc@cvut.cz} and Kristian Hengster-Movric\textsuperscript{b}}
\affil{\textsuperscript{a}Czech Institute of Informatics, Robotics, and Cybernetics, Czech Technical University in Prague, 160 00 Prague, Czech Republic \\ \textsuperscript{b}Department of Control Engineering, Faculty of Electrical Engineering, Czech, Czech Technical University in Prague, 166 36 Prague, Czech Republic}
}

\maketitle

\begin{abstract}
This study deals with analysis and control of cooperative synchronization for identical agents interacting on a directed graph topology. The agents are considered to have general continuous linear time-invariant dynamics with homogeneous communication and/or control delays. An LMI approach based on a Lyapunov-Krasovskii functional is proposed, together with the synchronizing region concept, which decouples the single-agent dynamics from the detailed graph topology. Moreover, the conventional notion of synchronizing region is here extended by an LMI relaxation utilizing quasi-convex characteristic of the problem. This leads to less conservative results for the region of graph matrix eigenvalues in the complex domain, where the synchronization is guaranteed. The proposed method to calculate the allowable delay bound for synchronization is also less conservative as compared to the results from the literature. % Additionally, the use of LMIs brings computational efficiency and ease. 
Furthermore, two designs for distributed state-feedback control are suggested. The precise delay value and the detailed graph topology need not be known for their application; it suffices only to know the upper bound on the delay and the approximate region where the Laplacian eigenvalues lie. Specific improvements over the results existing in the literature are demonstrated by a numerical example, which validates the proposed approaches.
\end{abstract}

% \end{frontmatter}
%%%%%%%%%%%%%%%
\section{Introduction}
The consensus/synchronization problem for multi-agent systems (MASs) has recently been a very attractive research field \citep{Caosurvey13,Qinsurvey17}. One of the main reasons for this is the broad area of its application to complex engineering tasks such as flocking/formation control of mobile robots and vehicles, see \citet{Ohsurvey2015}. Early works such as presented in \citet{Saber2004,Saber2007} focused on consensus problem for one-dimensional agents without a leader, where all agents converge to a final consensus value determined by their precise initial conditions. On the other hand, pinning a group of agents with general higher-order dynamics to a leader is considered in \cite{Wang2002} to synchronize all nodes to the leader node's trajectory for all initial conditions. Such synchronization problem, also called as {\emph{cooperative tracking}}, was studied for general linear time-invariant (LTI) agents in \cite{Fax2004,Li2011,Li2010} without pinning and in \cite{Zhang2011} for pinning control, utilizing Lyapunov stability theory. References \cite{Li2011,Li2010,Zhang2011} have considered decoupling the single-agent dynamics from the detailed network topology, through the concept of {\emph{synchronizing region}}, which greatly simplifies the design. The synchronizing region refers to the region in the complex plane for the graph matrix eigenvalues such that synchronization of a MAS is guaranteed. The idea of synchronizing region approach was first presented in \cite{Pecora98} for synchronization of identical coupled oscillators, where a master stability function is introduced to reduce effects of graph topology to robust stability of single-agent closed-loop dynamics. Due to this simplifying characteristic, the synchronizing region approach has been adapted for the solution of several MAS problems, e.g. optimal control in \cite{Hengster-Movric2014} and $H_\infty$ synchronization in \cite{Movric15aut}. 
% In \cite{Wang2002}, a distributed control scheme via pinning small number of agents to specific nodes, where the notion of pinning control is introduced.  
%
% The key idea in the synchronizing region approach is to reduce the global synchronization error dynamics of dimension $Nn$ to $N$ $n$-dimensional systems that depend on the eigenvalues of the Laplacian or pinned Laplacian matrix. With a coordinate transformation matrix, which leads to the transformed graph matrix in the upper triangular form, the global dynamics can be decomposed with respect to the graph topology, i.e. the eigenvalues of $L$ (or $L+G$), see \cite{Fax2004,Zhang2011}.
%
% in the concept of synchronizing region (c.f. Section \ref{sec:srwp}). (RECONSIDER THIS PART WITHOUT REFs)

Among the studies on synchronization of MASs, time-delay phenomena have been widely investigated since delays are unavoidable in communication links and quite common in control action. Thus, delays must be taken into account in analysis, and robustness to delays must be guaranteed for any realistic cooperative distributed control protocol. There are two main general approaches for analysis of time-delay systems (TDSs) \citep{gu}: frequency domain techniques and time-domain approach based on Lyapunov-Razumikhin functions (LRFs) or Lyapunov-Krasovskii functionals (LKFs). The first results based on Nyquist plots were given in \cite{Saber2004} for consensus of single-integrator agents on a fixed undirected graph topology, with a uniform delay. For the same problem on a digraph topology, a systematic approach in frequency domain
%, which introduces responsible eigenvalue concept 
for the analysis of delay margin guaranteeing consensus was proposed in \cite{Sipahi2011}, and a distributed robust consensus control was proposed in \cite{Lin2008} to achieve the desired $H_\infty$ performance. There are also more recent papers utilizing different frequency domain techniques for consensus of MASs with second-order agents on digraphs, e.g.~for double-integrator dynamics in \cite{Gomez2013SCL}, and for general second-order agents in \cite{Hou2017}. Both \cite{Gomez2013SCL} and \cite{Hou2017} consider decoupling/decomposition of the dynamics depending on eigenvalues of matrices describing the communication topology, e.g.~Laplacian eigenvalues in \cite{Hou2017}, so that the analysis can be performed with a set of low-order subsystems. Another remarkable work based on an inequality technique for second-order MASs with multiple time-varying delays is \cite{Zhu2010} where the decomposition of the global networked dynamics is considered.
%
%Another second order MAS work (nonuniform constant delays): \cite{Tian2009}
%Another work on single integrator dynamics (nonuniform time-varying delays): \cite{Sun2009}
%
Beside the MAS with low-order single-agents mentioned briefly above, several papers were published on consensus/synchronization problem of high-order linear agents with delays. In \cite{Munz2010,Munz2012}, robustness of consensus protocols to non-uniform constant feedback delays are investigated for identical and non-identical single-input single-output (SISO) agents with eigenvalues in the open left-half plane, networked on undirected topology. Consensus among multi-input multi-output (MIMO) identical agents on directed graphs with uniform constant communication delays has been studied in \cite{Wang2013stoor}, where the proposed approach is valid only for agents with eigenvalues in the closed left-half plane.~Frequency domain techniques utilized in \cite{Munz2010,Munz2012,Wang2013stoor} impose restrictions on single-agent dynamics such as requiring only SISO or only (marginally) stable agents. Predictor-based feedback protocols for consensus of MAS on directed graphs in the presence of communication and input delays have been proposed in \cite{Zhou2014}. Since predictor-based feedback suffers from implementation problems \citep{Zhou2012}, a truncated predictor approach has also been proposed in \cite{Zhou2014}, however it remains valid only for agents which are at most critically unstable just as in \cite{Wang2013stoor}. As for the time-domain methods, allowable uniform communication delay bound for consensus of general linear MASs on undirected networks, has been studied via an LKF approach in \cite{Zhang2014}. Based on LKFs, several papers such as \citep{Karimi2011,Karimi2012,Yi-Ping2015} worked on general synchronization under delays, however in a centralized manner  treating the entire MAS system as a whole. For identical agents on directed graphs with general LTI dynamics and uniform input delays, an algebraic Riccati equation (ARE) solution based on LRF has been proposed for cooperative synchronization in \cite{Movric15jfi}, with pinning control and synchronizing region approach. Very recently, allowable delay bound for the same problem without pinning on arbitrary network topology has been studied in \cite{Sheng2018}, however assuming the network has closed strong components. Another recent paper \cite{Zhang2018} proposes a parametric ARE-based solution for consensus of MAS with non-uniform constant communication delays on connected undirected graphs, where the tolerable delay bound is known, however the agents are assumed to be at most critically unstable. With the exception of \cite{Movric15jfi}, the proposed distributed methods in the mentioned papers impose restrictions on single-agent dynamics or on graph topology, even if the delays are uniform and constant. 
% Besides, several other difficulties/issues originating mostly from the utilized approaches arise in: (i) decoupling of the single-agent dynamics from the detailed communication topology, (ii) robustness to delay and graph topology uncertainties, (iii) conservatism both in allowable delay bound and the graph topology.
%Besides, most presently utilized approaches lead to difficulties for decoupling of the single-agent dynamics from the detailed communication topology, robustness to delay and graph topology uncertainties and reducing conservatism both in allowable delay bound and the graph topology.
%

In this paper, we consider synchronization of homogeneous general LTI agents with uniform constant communication (and/or control) delays. Agents are assumed interacting on a directed graph topology, thus general networks are considered in contrast to \cite{Munz2010,Zhang2014,Zhang2018} dealing only with undirected graphs. The underlying graph is assumed to have a directed spanning tree, which is a common and a mild assumption in the MAS literature. This is also a more general condition, compared to requiring the existence of closed strong components in \cite{Sheng2018}, where only the allowable delay bound is analyzed. There are no restrictions on agent dynamics as in \cite{Munz2010,Wang2013stoor,Zhou2014}, except for requiring single-agent stabilizability, which is inherently necessary for synchronization.  
Here we bring together, for the first time, the synchronizing region concepts of cooperative control, Lyapunov-Krasovskii delay-dependent stability analysis, and methods familiar from robust control to present a less conservative robust distributed control design. 
An LKF based LMI approach is thus proposed for both the analysis of cooperative synchronization in delayed MAS and the synchronization control design. 
Utilizing the well-known results of synchronizing region concept, we deal with robust stability of the single-agent closed-loop dynamics, where the uncertain parameters stand for the graph matrix eigenvalues and the delay, instead of dealing with the entire global networked dynamics, as e.g.~in \cite{Karimi2011,Karimi2012,Yi-Ping2015}. The proposed method relies on a newly revealed quasi-convexity property with respect to graph matrix eigenvalues, which is found to be crucial for our main development. Namely, based on quasi-convexity, a novel approach is presented for analyzing and treating the graph matrix eigenvalues. The proposed approach relaxes the conventional synchronizing region methodology in \cite{Movric15jfi}, which is found to be quite conservative both for delay bound and network topology requirements. This relaxation relies on using multiple LKFs and taking advantage of the quasi-convexity. In particular, we propose a unified approach for determining the allowable delay bound and the synchronizing region for any given state-feedback controller, and designing distributed cooperative state-feedback controls. Only the delay upper bound and the approximate region in the complex plane where the Laplacian eigenvalues should reside need to be known. 

The main contributions of this paper are summarized as follows: 
\begin{enumerate}
	\item A new approach to reveal the effects of network topology under communication/control delays is proposed in the spirit of synchronizing region. The synchronization criterion for delay bound and graph matrix eigenvalues is relaxed compared to the classical synchronizing region approach in \cite{Movric15jfi}, leading to considerably reduced conservatism. % by using multiple LKFs %as opposed to using single LRF in \cite{Movric15jfi}.
%The approach can be adapted to synchronization control problems for the delay-free case as well. 
	\item An efficient algorithm to estimate the delay-dependent synchronizing region is given based on feasibility of LMIs.
	\item A direct robust cooperative synchronization control design is provided without tuning a coupling gain, as opposed to most conventional synchronizing region designs, e.g. in \cite{Li2010,Zhang2011,Movric15jfi}. 
\end{enumerate}
% Building on existing results on robust control and TDSs, this paper brings together, for the first timewe provide a unified approach s is presented the proposed method brings a  enables to work on the single-agent dimensions approach brings utilizing the convexity 

The paper is organized as follows. Section \ref{sec:pre} provides notation, graph properties, and quasi-convexity preliminaries. Section \ref{sec:ps} presents the MAS dynamics and the problem statements. The proposed LKF based LMI approach to find the allowable delay bound is given in Section \ref{sec:adb}.~Based on the main results in Section \ref{sec:adb}, an efficient algorithm to obtain an estimate of the delay-dependent synchronizing region and two approaches for distributed state-feedback control design are given in Sections \ref{sec:sr} and \ref{sec:control}, respectively.~The results and improvements are verified via numerical simulations in Section~\ref{sec:ex}, and Section~\ref{sec:conc} concludes the paper. 
\section{Preliminaries} \label{sec:pre}

 Throughout the paper, $\rr^n$ denotes the n-dimensional real Euclidean space, and $\rr^{n\times m}$ is the set of all $n\times m$ matrices. $\cc$ stands for the complex numbers. For any complex number $z$, we denote its conjugate by $z^*$. 
The notation $P\succ 0$, for $P\in\rr^{n\times n}$ means that $P$ is symmetric and positive definite.~The symmetric elements of symmetric (or hermitian) matrices are denoted by $\star$. The notation $P^H$ stands for the hermitian (or conjugate) transpose of the matrix $P$. The space of functions $\phi : [-h,0]\rightarrow \rr^n$, which are absolutely continuous on $[-h,0]$, and have square integrable first order derivatives, i.e. $\dot{\phi}\in L_2(-h,0)$, is denoted by $W[-h,0]$ with the norm $\| \phi \|_W = {\text{max}}_{\theta \in [-h,0]}|\phi(\theta)| + [\textstyle\int_{-h}^0 |\dot\phi(\theta)|^2 ds]^{1/2}$. For $x:\rr\rightarrow\rr^{n}$, we denote $x_t(\theta)\triangleq x(t+\theta), \; \theta \in [-h,0]$. 
\subsection{Graph Theory \& Properties}
$\mathcal{G}=(\mathcal{V},\mathcal{E})$ represents a graph with $N$ nodes $\mathcal{V}=\{v_1,\ldots,v_N \}$ and a set of edges $\mathcal{E}\subseteq\mathcal{V}\times\mathcal{V}$. An edge rooted at $v_j$ and ended at $v_i$ is denoted by ($v_j,v_i$). We consider only simple graphs, i.e. there are no repeated edges or self-loops. A sequence of successive edges from $v_i$ to $v_j$ in the form $\{(v_i,v_k), (v_k,v_l), \ldots, (v_m,v_j) \}$ is called as {\emph{directed path}}. A graph is {\emph{strongly connected}} if every two nodes can be joined by a directed path. The graph is said to contain a directed {\emph{spanning tree}} if there exists a node $v_0$ such that every other node in $\mathcal{V}$ can be connected to $v_0$ by a directed path starting from $v_0$. Such a special node is then called a {\emph{root node}} denoted by $v_{i_r}$. 

{\emph{Adjacency (connectivity) matrix}} is denoted as $E = \{ \alpha_{ij} \}$, with $\alpha_{ij}>0$ if $(v_j,v_i)\in\mathcal{E}$ and $\alpha_{ij}=0$ otherwise. The set of neighbors of node $v_i$ is denoted as $\mathcal{N}_i = \{v_j: (v_j,v_i) \in\mathcal{E}\}$, i.e. the set of nodes with edges coming into $v_i$. $D = \text{diag}(d_1,\ldots,d_N)$ with $d_i = \textstyle{\sum}\alpha_{ij}$ the weighted degree of node $i$ denotes the in-degree matrix for $\mathcal{G}$. Then the {\emph{(graph) Laplacian matrix}} is defined as $L = D-E$. The Laplacian matrix $L$ has a simple zero eigenvalue if and only if the digraph contains a spanning tree. %the undirected graph is connected, or 
\subsection{Convex Polytopes and Quasi-convexity}
% The $M$-unit simplex is denoted by $\Lambda_M$. 
A compact and convex polytope, denoted by $\Delta_M$, is characterized also by the set of its vertices, $V:= \{ \nu_1, \ldots, \nu_M\}$. We denote the operation of taking the convex hull by ${\mathbf{co}}\{\cdot\}$, and the operation of taking the set of vertices by ${\mathbf{vert}} \{ \cdot \}$. So, by definition $\Delta_M = {\mathbf{co}}\{V\}$ and $V = {\mathbf{vert}} \{ \Delta_M \}$. Then we have the following. 
%Let us also represent the following definition and theorem on quasi-convexity. 
%
\begin{definition}[Quasi-convexity, \cite{Apkarian2000}]
Let $S$ be a convex subset of $\rr^n$. A function $f: S\rightarrow \rr$ is quasi-convex if for all $u,v$ in $S$ and $\alpha\in[0,1]$, $f(\alpha u + (1-\alpha)v) \leq {\rm{max}} \{ f(u) , f(v) \}$. 
\label{def:qc}
\end{definition}
\begin{theorem}[Multi-quasi-convexity, \cite{Apkarian2000}] \label{th:mqc}
Consider a compact convex polytope $\Delta_M$ in $\rr^n$ and the directions $b_1,\ldots,b_q$ determined by the edges of $\Delta_M$. Assume that for any $x$ in $\Delta_M$, the function $f$ is quasi-convex on the line segments 
\begin{equation*}
L_{b_i}(x) := \{ z\in \Delta_M : z = x + \xi b_i, \xi \in \rr \}, i=1,\ldots,q.
%\label{eq:}
\end{equation*}
% $L_{d_i}(x)$ for $i=1,\ldots,q$. 
Then, $f$ has a maximum over $\Delta_M$ at a vertex of $\Delta_M$. 
\end{theorem}
%
%%%%%%%%%%%
\section{Multi-agent System and Problem Statement} \label{sec:ps}
We consider a set of $N$ homogenous agents having general LTI dynamics with control input delay, $\tau_{\rm{in}}\geq 0$
% (identical) (nodes)
\begin{equation}
\dot{x}_i(t)=A  x_i(t)+B u_i(t-\tau_{\rm{in}}), \quad i=1,\ldots,N,
\label{eq:sadyn}
\end{equation}
and a leader system with autonomous dynamics % (node 0)
\begin{equation}
\dot{ x}_0(t)=A  x_0(t)
\label{eq:ldyn}
\end{equation}
where $ x_0, x_i\in\rr^n$, $ u_i\in\rr^m$ are states and inputs, respectively, and $A\in\rr^{n\times n}$ and $B\in\rr^{n\times m}$. One can consider the leader system (\ref{eq:ldyn}) as an exosystem serving as a command generator  \citep{Zhang2011}, which is to be tracked by all the agents. For that purpose, the leader pins to a subset of agents. 

The {\emph{distributed cooperative tracking (or synchronization)}} control goal is to find $u_i$ depending on $x_i, \; x_{j\in\mathcal{N}_i}$ such that all agents' states asymptotically synchronize to the trajectory of the leader node (\ref{eq:ldyn}), i.e.~$\parallel x_i(t) - x_0(t)\parallel \rightarrow 0$ as $t\rightarrow\infty$. We consider the distributed state-feedback control of the form 
%
%The leader pins to a subset of agents, and it is known that the agents' states synchronize to the trajectory of the leader node (\ref{eq:ldyn}), i.e. $\parallel x_i(t) - x_0(t)\parallel \rightarrow 0$ as $t\rightarrow\infty$, with an appropriate distributed (pinning) control law \cite{Zhang2011}. To find such controls $u_i$ depending on $x_i, \; x_{j\in\mathcal{N}_i}$ is called as {\em{distributed cooperative tracking (or synchronization)}} control problem. %pinning control in \cite{Li2010} 
%For this control purpose, we investigate the local state-feedback control law
\begin{equation}
 u_i(t)= K e_i(t),
\label{eq:fb}
\end{equation}
applied to each agent, $\forall i$,  with the {\emph{delayed local neighborhood error}}
\begin{equation}
e_i (t)=\sum_j{\alpha_{ij}\left[ x_j(t-\tau_{\rm{c}})- x_i(t-\tau_{\rm{c}})\right]+g_i\left[ x_0(t-\tau_{\rm{c}})- x_i(t-\tau_{\rm{c}})\right]}
\label{eq:locerr}
\end{equation}
where $\alpha_{ij}$ are the elements of adjacency matrix $E$ of the directed graph topology of the MAS, $g_i \geq 0$ are the pinning gains that are nonzero only for a few agents directly connected to the leader, and $\tau_{\rm{c}}\geq 0$ represents the communication delay. 
Then, the closed-loop dynamics of each agent with the control (\ref{eq:fb}) becomes
\begin{equation}\label{eq:clsingle}
\dot{x}_i(t)=A x_i(t)+ B K e_i(t-\tau_{\rm{in}}).
\end{equation}
%
% \begin{remark}
% Note that one can also consider delayed control input in (\ref{eq:sadyn}), i.e. $\bar u_i(t) \rightarrow \bar u_i(t-\tau_c)$, together with the communication delay in (\ref{eq:locerr}), where the overall delay is to be considered as $\tau + \tau_c$ in (\ref{eq:clsingle}). 
% \end{remark}

For future LKF analysis, we assume that the initial conditions of the MAS (\ref{eq:sadyn}) are $x_{it}(\theta)\in\phi$ and $u_{it}(\theta)\in\phi$, $\theta\in[-\tau,0]$, where $\tau = \tau_{\rm{in}} + \tau_{\rm{c}}$ is the total delay. Defining the synchronization error of an agent as $\delta_i=x_i-x_0$ and considering the total delay $\tau$, %$\tau = \tau_{\rm{in}} + \tau_{\rm{c}}$,
the global synchronization error dynamics of the networked system can be written straightforwardly as  
\begin{equation}
\dot{\delta}(t)= \left( I_N \otimes A \right) \delta (t) - (L+G)\otimes B K \delta (t-\tau),
\label{eq:glerr}
\end{equation}
where $\delta (t) = \left[\delta_1^T (t) \ldots \delta_N^T (t) \right]^T \in \rr^{N\cdot n}$, and $G={\rm{diag}}(g_1,\ldots, g_N)$ is the {\emph{pinning gain matrix}}. Notice that the {\emph{pinned (graph) Laplacian matrix}} ($L+G$) appears in synchronization dynamics (\ref{eq:glerr}), instead of the Laplacian matrix ($L$) itself. Let us here present the assumption on the graph topology that holds throughout the paper, and a relevant lemma on the pinned Laplacian matrix which is important and well-known. 
%Throughout the paper, digraphs are considered, and the following assumption holds for the graph topology. 
\begin{assumption}\label{assume1}
The digraph $\mathcal{G}$ contains a directed spanning tree with at least one nonzero pinning gain into a root node. 
\end{assumption}

\begin{lemma}\citep{Wang2002,Zhang2011} \label{lem:lapeig}
Under Assumption \ref{assume1}, $L+G$ is non-singular with all the eigenvalues ($\lambda_j$) in the open right-half complex plane, i.e. $\lambda_j \in \cc_{+}$, $j=1,\ldots$,N. 
\end{lemma}

For the synchronization of the MAS (\ref{eq:sadyn}--\ref{eq:locerr}), the global synchronization error dynamics (\ref{eq:glerr}) is required to be asymptotically stable, i.e. $\lim_{t\to\infty} \delta(t) \rightarrow 0 $.~Thus, we state the considered problems as follows: 

\begin{problem}\label{prob1}
Find the allowable delay bound ($\tau\leq h$) in (\ref{eq:glerr}) such that the synchronization of the MAS (\ref{eq:sadyn}--\ref{eq:locerr}) is  guaranteed for all $\tau \in [0,h]$, for a given specific $K$. 
\end{problem}
\begin{problem}\label{prob2}
Find the region for the eigenvalues of $L+G$ in the complex plane where the MAS synchronization with (\ref{eq:sadyn}--\ref{eq:locerr}) is guaranteed for all $\tau \in [0,h]$ and for a given specific $K$. 
\end{problem}
\begin{problem}\label{prob3}
Provide a design method, namely find a distributed cooperative state-feedback control law (\ref{eq:fb}) with (\ref{eq:locerr}) $\forall\tau \in [0,h]$ such that the states of the agents (\ref{eq:sadyn}) synchronize to the states of the leader (\ref{eq:ldyn}), for a given delay bound $(h)$ and for a given region of the eigenvalues of $L+G$.
\end{problem}
\begin{remark}\label{rem1}
We assume that the total delay $\tau$ is constant and uniform for each communication link. Even though this assumption usually does not fit the most realistic cases, nevertheless it is a general assumption in the literature, see e.g. \cite{Saber2004,Zhou2014,Yi-Ping2015,Movric15jfi,Sheng2018,Wang2016}, allowing concise theoretical developments, in particular focusing on single-agent dynamics instead of the entire MAS. 
Moreover, from the practical side, even if the delays are originally found to be unequal, they can be made uniform through specific implementation by using time-stamps and inserting artificial delays on the signals received from neighboring agents, to guarantee the results presented in this work remain valid. 
% Also, the delays may not even be uniform in realistic implementations, it is possible to make uniform all delays for each communication links, using the time-stamps and inserting artificial delays on the signals received from any agent, to guarantee the results presented in this work is valid.
\end{remark}

Note that (\ref{eq:glerr}) is a retarded functional differential equation (RFDE), which has infinitely many characteristic roots. The stability and stabilization problem of TDSs is known to be NP-hard \citep{gu}. As the number of agents $N$ increases, i.e.~the dimension of RFDE (\ref{eq:glerr}) gets higher, the computational load for the analysis/synthesis problem of the TDS becomes cumbersome. This is so on large scale networks even for the delay-free case. 
%Furthermore, it is more convenient to consider lower order dynamics for the distributed control design. 
%decoupling the agent dynamics from the network topology features provides insight to investigate the effect of the graph topology or to provide control design regarding the network structure. 
To avoid such difficulties and to reveal the effects of graph topology on the global dynamics, we resort to synchronizing region approach \citep{Li2011,Li2010,Zhang2011}. 
%the {\em{synchronizing region approach}}, which also constructs the basis of this study, has been developed. We provide an overview and some useful results on that approach in the following subsection. 
%
% Obviously the solution to Problem \ref{prob3} is the solution to the cooperative tracking control problem.

%%%
% \subsection{Synchronizing Region Approach with Pinning Control} \label{sec:srwp}% An Overview of the 

% The idea of synchronizing region approach was first presented in \cite{Pecora98} for the synchronization of identical coupled oscillators where a master stability function is introduced to reduce the effects of graph topology to the robust stability of single-agent closed loop dynamics. In \cite{Wang2002}, a distributed control scheme via pinning small number of agents to specific nodes, where the notion of pinning control is introduced.  Based on Lyapunov stability theory together with the synchronizing region and pinning control approaches, the distributed state-feedback design for synchronization has been considered via LMIs in \cite{Li2011,Li2010} and via local ARE design in \cite{Zhang2011}. 

The key idea in the synchronizing region approach introduced below is to reduce the global synchronization error dynamics of dimension $Nn$ to $N$ $n$-dimensional systems that depend on the eigenvalues of the graph matrix. With a coordinate transformation matrix, which leads to the transformed graph matrix in the upper triangular form, the global dynamics can be decomposed with respect to the graph topology, i.e.~the eigenvalues of $L$ (or $L+G$), see \cite{Fax2004,Zhang2011}. This decomposition was extended to the leader follower MAS synchronization problem with uniform constant delay in \cite{Movric15jfi} as in the following lemma, under Assumption \ref{assume1}.
\begin{lemma} \citep{Movric15jfi} \label{lem:decom}
The trivial solution of the system (\ref{eq:glerr}) is asymptotically stable if and only if the trivial solutions of the systems 
\begin{equation}
\dot{x}(t)=A x(t)-\lambda_j B K x(t-\tau), \; j=1,\ldots, N
\label{eq:lem1}
\end{equation}
each having the order of a single-agent,~where $\lambda_j$ are the eigenvalues of L+G,~are asymptotically stable $\forall j$. 
\end{lemma}

As a result of Lemma \ref{lem:decom} for the synchronization problem, instead of analyzing (\ref{eq:glerr}) with higher dimension ($N n$), one can investigate the robust stability of a lower dimension ($n$) system (\ref{eq:lem1}). This also provides an insight into graph topology both for analysis and design. 

Note that (\ref{eq:lem1}) is a complex RFDE in $\cc^n$ depending on $\lambda_j$. 
% Let us recall Lemma \ref{lem:lapeig} that the $\lambda_j$'s are complex numbers, so (\ref{eq:lem1}) describes a dynamical system generally in $\cc^n$.
Thus, to provide a convenient representation from the perspective of robustness, let us rewrite (\ref{eq:lem1}) in the form, 
\begin{equation}
\dot{x}(t)=A x(t) + \sigma A_d x(t-\tau),
\label{eq:crfde}
\end{equation}
where $A_d:=-BK$, and $\sigma\in \cc$. Then, the following definition is given to investigate the robust stability of (\ref{eq:crfde}) with respect to $\sigma$ and $\tau$. 
\begin{definition} \label{def:synch}
Given a complex RFDE as in (\ref{eq:crfde}), the delay-dependent synchronizing region (DSR) is a subset of the complex plane $\cc$ depending on the delay bound
\begin{equation}
S_c(h)=\left\{\sigma\in \cc : A x(t) + \sigma A_d x(t-\tau) {\rm{\; is \; asymptotically \; stable \;}} \forall\tau\in[0,h] \right\}.
\label{eq:sr}
\end{equation}
\end{definition}
Note that Definition \ref{def:synch} slightly differs from the one in \cite{Movric15jfi}, since we consider the synchronizing region where the robust stability for any $\tau\in[0,h]$ is guaranteed.~Definitely, the DSR depends also on the system dynamics ($A,B$) and the local feedback matrix ($K$). The use of $\sigma$ is also for the consideration of scaling the pinned Laplacian eigenvalues with a coupling constant $c>0$. This scaling approach will be considered later (c.f. Section \ref{sec:control2}) for control design with $ K = c \bar{K}$, where we result in $c \lambda_j \in S_c(h)$.

%%%%%%%%%%%%%%%%%%%%%%%%%%%%%

\section{Allowable Delay Bound for Synchronization} \label{sec:adb}
This section presents a solution to Problem \ref{prob1}, and provides a basis for the subsequent solutions of Problems \ref{prob2} and \ref{prob3}. For this purpose, we utilize the LKF based approach, which is known to reduce conservatism in analysis and control of TDSs as compared to the LRF approach \citep{Fridman2014book,Briat2015book}. This is due to a more general structure of LKF depending also on derivatives of the system states. 
Let us first present the Lyapunov-Krasovskii stability criterion for general TDSs, which is given in \cite{Fridman2014} with a slight modification as compared to \cite{gu}.% , for the functional depending on the state derivatives, i.e. $V(t,x_t,\dot x_t)$. 
\begin{theorem} [Lyapunov-Krasovskii Stability, \cite{Fridman2014}] \label{eq:kra1} Consider the TDS:
\begin{equation}
\dot{x}(t)=f(t,x_t), \; t\geq t_0,
\label{eq:tds}
\end{equation}
where $f:\rr \times\mathcal{C}[-h,0]\rightarrow \rr^n$ is continuous in both arguments and is locally Lipschitz continuous in the second argument. Suppose $f$ maps $\rr\times$ (bounded sets in $\mathcal{C}[-h,0]$) into bounded sets of $\rr^n$ and that $u,v,w:\rr_{\geq t_0}\rightarrow \rr_{\geq t_0}$ are continuous non-decreasing functions, $u(s)$ and $v(s)$ are positive for $s>0$, and $u(0)=v(0)=0$. Assuming $f(t,0)=0$, the trivial solution of (\ref{eq:tds}), i.e. $x(t)\equiv 0$, is uniformly stable, if there exists a continuous functional $V:\rr\times W[-h,0] \times L_2(-h,0) \rightarrow \rr^{+} $, which is positive-definite, i.e.
\begin{equation}
u(|x(t)|)\leq V(t,x_t,\dot{x}_t)\leq v(\parallel x_t \parallel_W),
\end{equation}
and such that its derivative along (\ref{eq:tds}) is non-positive in the sense that 
\begin{equation}
\dot{V}(t,x_t,\dot{x}_t)\leq -w(|x(t)|) .
\label{eq:kra2}
\end{equation}
If $w(s)>0$ for $s>0$, then the trivial solution is uniformly asymptotically stable. If in addition $\lim_{s\to\infty} u(s)=\infty$, then the uniform asymptotic stability is global. 
\label{th:lk}
\end{theorem}

Note that, there are various model transformation and bounding techniques applied both in LRF and LKF approaches, which usually give conservative estimates of the upper delay bound $h$ for which the stability is guaranteed. One can read \cite{Briat2015book} for a detailed overview on these techniques. For comparison, note that the conservatism in \cite{Movric15jfi} comes also from an approximate model transformation as well as the nature of LRF approach itself. In contrast, here the utilized LKF approach does not require a model transformation, and relies on a bounding technique given in the following lemma, which is commonly used and gives satisfactory reduced conservatism for RFDEs with constant delays \citep{Fridman2014book}.

\begin{lemma}[Jensen Inequality, \citep{Fridman2014,gu2000}] \label{lem:jensen}For any constant matrix $R \in \rr^{n\times n}$, and any function $\phi\in L_2[-h,0]$, the following inequality holds: %(the space of square integrable functions)
\begin{equation}
\int_{-h}^0 \phi^T(s) R \phi(s) ds \geq \frac{1}{h} \left( \int_{-h}^0 \phi^T(s) ds \right) R \left( \int_{-h}^0 \phi(s) ds \right).
\label{eq:jensen}
\end{equation}
\end{lemma}
% \vspace{-12pt}
%
%Now, we can give the main theorem based on the extension of the LKF given in  \cite{Gouaisbaut2006} to the stability of a complex RFDE given by (\ref{eq:crfde}). %following the lines of \cite{Fridman2014}. 

Now, following the lines in \cite{Fridman2014book,Fridman2014}, we can give the main theorem on the stability of the complex RFDE (\ref{eq:crfde}) based on a commonly used LKF. 
\begin{theorem} [Delay-dependent stability for the complex RFDE] \label{th:main}
Given $h\geq 0$ and $\sigma\in\cc$, let there exists $n\times n$ matrices $P\succ 0, \;S\succ 0, \;R\succ 0$, such that the following LMI is feasible:
\begin{equation}
\begin{bmatrix}
    A^T P+PA+S-R     & \sigma P A_d+R  & h A^T R \\
    \star             & -S-R   & \sigma^{*} h A_d^T R \\
		\star             & \star   & -R \\
\end{bmatrix} \prec 0.
\label{eq:lmi1cons} 
\end{equation} 
Then, the system (\ref{eq:crfde}) is uniformly asymptotically stable for all delays $\tau \in [0,h]$. 
\label{th:crfde}
\end{theorem}
\begin{proof}
Consider the LKF
\begin{equation}
V (x, \dot x, t) =x^H(t) P x(t) + \int_{t-h}^{t}  x^H(s) S x(s)ds + h \int_{-h}^{0} \int_{t+\theta}^{t} \dot x^H(s) R \dot x(s)ds d\theta
\label{eq:lkfcons}
\end{equation}
where $P\succ 0, S\succ 0, R\succ 0$. Notice that in (\ref{eq:lkfcons}), $V:\rr\times W[-h,0] \times L_2(-h,0) \rightarrow \rr^{+}$ as in Theorem \ref{th:lk}, is real valued, even if the states (and their derivatives) are in a complex vector space. %, since the hermitian transposition results in a real valued functional with the quadratic form and the real symmetric positive-definite matrices. 
Then, (\ref{eq:lkfcons}) satisfies condition (\ref{eq:kra1}) considering the positive continuous non-decreasing functions $u(x)=x^H(t) P x(t)$ and $v(x, \dot x, t)=V (x, \dot x, t) |_{\tau=h}$, as in Theorem \ref{th:lk}. 

The double integral term in (\ref{eq:lkfcons}) can be rewritten \citep{Fridman2014} as
\begin{equation}
h \int_{-h}^{0} \int_{t+\theta}^{t} \dot x^H(s) R \dot x(s)ds d\theta = h \int_{t-h}^{t} (h+s-t) \dot x^H(s) R \dot x(s) ds, \; R \succ 0 . 
\label{eq:ditcons}
\end{equation}
Considering (\ref{eq:ditcons}), the derivative of the functional (\ref{eq:lkfcons}) is obtained as
\begin{equation}
\begin{split}
\frac{dV}{dt} = &  \; \dot x^H(t) P x(t) +  x^H(t) P \dot x(t) + x^H(t) S x(t) - x^H(t-h) S x(t-h) \\ 
& + h^2 \dot{x}^H (t) R \dot{x} (t) - h \int_{t-h}^{t} \dot x^H(s) R \dot x(s)ds.
\label{eq:dlkfcons}
\end{split}
\end{equation}
Now, applying Jensen's inequality (\ref{eq:jensen}) to the integral term in (\ref{eq:dlkfcons}), we can write the inequality 
\begin{equation}
\begin{split}
\frac{dV}{dt} \leq \;  \dot x^H(t) P x(t)  & + x^H(t) P \dot x(t) + x^H(t) S x(t) - x^H(t-h) S x(t-h) + h^2 \dot x^H (t) R \dot{x} (t) \\
& - \underbrace{\int_{t-h}^{t} \dot x^H(s) ds R \int_{t-h}^{t} \dot x(s)ds}_{=[x(t) \; -x(t-h)]^H R [x(t) \; -x(t-h)]} . 
\label{eq:dlkfineqcons}
\end{split}
\end{equation}
Denote $\eta(t) = [x(t) \quad x(t-h)]^T$. Then, substituting for $\dot x(t)$ the right-hand side of (\ref{eq:crfde}) into (\ref{eq:dlkfineqcons}), we get
\begin{equation}
\frac{dV}{dt} \leq \eta^H(t) \left[ \Phi  +  (RW)^H R^{-1} (RW) \right] \eta (t)
\label{eq:dlkfineq2cons}
\end{equation}
where
\begin{equation}
\Phi = \begin{bmatrix}
    A^TP+PA+S-R         & \sigma P A_d + R \\
    \star                & -S-R   	\\
\end{bmatrix} , W = h \left[ A \quad \sigma A_d\right]. 
\label{eq:phiwcons}
\end{equation}
Applying the Schur complement to the matrix in (\ref{eq:dlkfineq2cons}) yields the matrix on the left-hand side of (\ref{eq:lmi1cons}). If  LMI (\ref{eq:lmi1cons}) is feasible, then $dV /dt \leq -\varepsilon |x(t)|^2$ for some $\varepsilon > 0$, which guarantees the asymptotic stability of the system (\ref{eq:crfde}). 
\end{proof}

%\begin{remark}
%Note that if $h\rightarrow 0$, the LKF (\ref{eq:lkfcons}) reduces to a Lyapunov function of the kind conventionally used in the synchronizing region approaches, e.g. in \cite{Li2010,Zhang2011}. \ba{However, I couldn't see the straightforward conventional Riccati design from $\Phi$ in (\ref{eq:phiwcons}), since also $S\rightarrow 0, R\rightarrow 0$ reduces. I guess we need more than a remark. }
%\label{remdelayfree}
%\end{remark}
%
LMI (\ref{eq:lmi1cons}) in $P,S,R$ matrix variables is complex valued. Its left-hand side, denote as $M$, is Hermitian. $M \prec 0$ can readily be handled by its real symmetric matrix form  as \citep{Boyd93}:
\begin{equation}
%L(\eta) \prec  0 \Leftrightarrow  \tilde{L}(\eta) \triangleq \begin{bmatrix} \Re(L(\eta)) & -\Im(L(\eta)) \\ \Im(L(\eta)) & \Re(L(\eta)) \end{bmatrix} \prec  0. 
M \prec  0 \Leftrightarrow  \tilde M \triangleq \begin{bmatrix} \Re(M) & -\Im(M) \\ \Im(M) & \Re(M) \end{bmatrix} \prec  0. 
\label{eq:imlmi}
\end{equation}
For robustness analysis, let us denote $\gamma:=\Re(\sigma)$ and $\beta:=\Im(\sigma)$. Then, (\ref{eq:lmi1cons}) can be written as 
\begin{equation}
\begin{bmatrix}
A^TP+PA+S-R & \gamma PA_d+R & h A^TR & 0 & \beta P A_d & 0 \\
\star & -S-R & \gamma h A_d^T R   & -\beta A_d^T P & 0 & -\beta h A_d^T R \\
\star & \star & -R & 0 & \beta h R A_d & 0 \\
\star & \star & \star & A^TP+PA+S-R & \gamma PA_d+R & h A^TR \\
\star & \star & \star & \star & -S-R & \gamma h A_d^T R \\
\star & \star & \star & \star & \star & -R
\end{bmatrix} \prec 0,
\label{eq:lmi1rep}
\end{equation}
using (\ref{eq:imlmi}). Note that (\ref{eq:lmi1rep}) is not an LMI when $\gamma$ and $\beta$ are considered as variables, even if $h$ is fixed, due to their multiplications with matrix variables $P$ and $R$. However, it is still affine in both $\gamma$ and $\beta$, if $h$ is fixed. 
% Since affine functions are multi-quasi-convex ,
% Then, the feasibility of (\ref{eq:lmi1rep}) constitutes a generalized eigenvalue problem, which is a quasi-convex optimization problem in $\gamma$ ad $\beta$ for fixed $h$ \cite{Boyd1994}. 
%
With this observation on (\ref{eq:lmi1cons}), we give the following theorem for the robust stability of (\ref{eq:crfde}) with respect to uncertainty in $\sigma$ .
%we give the following theorem by the immediate application of Theorem \ref{th:mqc} to Theorem \ref{th:main}. 
\begin{theorem}[Robust Stability of the complex RFDE] \label{th:robust}
Given $h\geq 0$, if the LMI (\ref{eq:lmi1cons}) is feasible for a point set $\Sigma_M = \{ \sigma_l\in\cc\}$, $l=1,\ldots,M$, with some $P_l\succ 0, S_l\succ 0, R_l\succ 0$, then the system (\ref{eq:crfde}) is uniformly asymptotically stable, $\forall\sigma\in{\mathbf{co}}\{\Sigma_M\}$ and $\forall\tau \in [0,h]$.
\end{theorem}
\begin{proof}
Representing (\ref{eq:lmi1cons}) as (\ref{eq:lmi1rep}) for a given $h\geq 0$, 
the feasibility of (\ref{eq:lmi1cons}) depending on $\gamma$ and $\beta$ can be expressed as a generalized eigenvalue problem \citep{Boyd1994}, in the form of
\begin{equation*}
\begin{array}{cc}
    \rm{maximize} 					& \quad \gamma \; ({\rm{or}} \; \beta) \\
    \rm{subject \; to}      & (\ref{eq:lmi1rep}), \quad P\succ 0, \;S\succ 0, \;R\succ 0.\\
		\end{array} 
\label{eq:gevp}
\end{equation*}
Since the left-hand side of (\ref{eq:lmi1rep}) is affine in both $\gamma$ and $\beta$, it is trivially quasi-convex for all ($\gamma,\beta$) in a convex subset of $\rr^2$ (see Definition \ref{def:qc}). Thus, by Theorem \ref{th:mqc}, it has its maximum at a vertex of ${\mathbf{co}}\{\Sigma_M\}$, where $\Sigma_M = \{ \sigma_l\in\cc\}$, $l=1,\ldots,M$, describes a complex point set. Consequently, if (\ref{eq:lmi1cons}) is feasible for all vertices of ${\mathbf{co}}\{\Sigma_M\}$ with some $P_l\succ0,S_l\succ0,R_l\succ0$, there also exists feasible solutions to LMI (\ref{eq:lmi1cons}), $\forall\sigma\in{\mathbf{co}}\{\Sigma_M\}$. This guarantees the uniform asymptotic stability of (\ref{eq:crfde}) by Theorem \ref{th:main}, for any $\sigma\in{\mathbf{co}}\{\Sigma_M\}$, which completes the proof.
%
%Representing (\ref{eq:lmi1cons}) as (\ref{eq:lmi1rep}) for given $h\geq 0$, the feasibility of (\ref{eq:lmi1cons}) depending on $\gamma$ and $\beta$ can be expressed in the form of a generalized eigenvalue problem, where the objective function is quasi-convex \cite{Boyd1994}, subject to $P,S,R$. 
%% Then, the feasibility of (\ref{eq:lmi1rep}) constitutes a generalized eigenvalue problem, which is a quasi-convex optimization problem in $\gamma$ ad $\beta$ for fixed $h$ \cite{Boyd1994}.
%%Since any LMI is also affine in its decision variables and affine functions are trivially quasi-convex (see Definition \ref{def:qc}), the determinant of the left hand side of (\ref{eq:lmi1cons}) has maximum at a vertex of 
%Then, by Theorem \ref{th:mqc}, the objective function has its maximum at a vertex of ${\mathbf{co}}\{\Sigma_M\}$. Consequently, if (\ref{eq:lmi1cons}) is feasible for all vertices of ${\mathbf{co}}\{\Sigma_M\}$ with some $P_l\succ0,S_l\succ0,R_l\succ0$, there also exists feasible solutions to LMI (\ref{eq:lmi1cons}) $\forall\sigma\in{\mathbf{co}}\{\Sigma_M\}$. This guarantees the uniform asymptotic stability of (\ref{eq:crfde}) by Theorem \ref{th:main}, for any $\sigma\in{\mathbf{co}}\{\Sigma_M\}$, which completes the proof.
\end{proof}
%
%
%\begin{proof}
%Immediate by the application of Theorem \ref{th:mqc} to Theorem \ref{th:main}, where  the polygon ${\mathbf{co}}\{\Lambda_M\}$
%\end{proof}
%Thus, its explicit form for a fixed $h$ can be represented as
%\begin{equation}
%\tilde{M}_0(\eta) + \gamma \tilde{M}_1(\eta) + \beta \tilde{M}_2(\eta) \prec 0
%\label{eq:lmi1exp}
%\end{equation}
%where $\eta$ denotes the decision variable vector, and $M_0(\cdot), M_1(\cdot), M_2(\cdot)$ are real symmetric matrix-valued and linear functions of $\eta$.\\
% CONSIDER HOW TO FORMULATE AS GEVP! RELATE MULTI-AFFINITY AND MULTI-QUASI-CONVEXITY... 
%
%affine in $h$, when $\sigma$ is known. Its left-hand side is a complex Hermitian matrix, which are readily handled by the representation. 
%Now, let us reveal an important property of the LMI (\ref{eq:lmi1cons}), which is dependent on $\sigma$ and $h$. Notice that $h$ is  the quasi-convexity property of the LMI (\ref{eq:lmi1cons}) for the stability of the complex RFDE (\ref{eq:crfde}). 

% The following theorem is given with the quasi-convexity property of the LMI (\ref{eq:lmi1cons}) depending on $\sigma$ for the allowable (admissible) delay bound where the synchronization is guaranteed. 
The following result is a specific application of Theorem \ref{th:robust}, which provides the upper delay bound for the MAS synchronization (Problem \ref{prob1}), which is one of the main results of this paper.
\begin{theorem} [Delay bound for synchronization] 
\label{th:delay}
Let $\Lambda_N = \{\lambda_j\in\cc\}$, $j=1,\ldots,N$ be the set of all eigenvalues of L+G. % regarding a fixed graph topology of MAS with (\ref{eq:sadyn}-\ref{eq:locerr}). 
Define its convex hull as $\Delta:={\mathbf{co}}\{\Lambda_N\}$, and its corresponding set of vertices as $V_M:={\mathbf{vert}}\{\Delta\}$. 
Given $h\geq 0$ and $A_d = -BK$, let there exists $P_l\succ 0, S_l\succ 0,R_l\succ 0$, such that LMI (\ref{eq:lmi1cons}) with $\sigma\leftarrow\nu_l$, $\forall \nu_l\in V_M$, $l=1,\ldots,M$, is feasible.
Then, the control law given by (\ref{eq:fb}) and (\ref{eq:locerr}) asymptotically synchronizes all agents (\ref{eq:sadyn})  to the leader node (\ref{eq:ldyn}) for all delays $\tau \in [0,h]$. 
\end{theorem}
\begin{proof}
Given $h\geq 0$ and $A_d = -BK$, apply Theorem \ref{th:robust} to the complex RFDE (\ref{eq:lem1}) with $\sigma\leftarrow\lambda_j$, where $\lambda_j$'s are all the eigenvalues of graph matrix $L+G$. Then, (\ref{eq:lem1}) is uniformly asymptotically stable for all $\lambda_j$'s, if (\ref{eq:lmi1cons}) is feasible for all $\nu_l\in V_M={\mathbf{vert}}\{{\mathbf{co}}\{\Lambda_N\}\}$. Consequently, by Lemma \ref{lem:decom}, the global synchronization error system (\ref{eq:glerr}) is asymptotically stable,~i.e. $\lim_{t\to\infty} \delta(t) \rightarrow 0 $.~So,~the synchronization of the states of (\ref{eq:sadyn}) to the states of leader node (\ref{eq:ldyn}) is guaranteed, $\forall\tau\in[0,h]$. 
%{\textbf{Maybe better to use the proof in Theorem 2.4.2 in \cite{Briat2015book}, however it seems like it is for common LF, see also \cite{Gouaisbaut2006}. On the other hand, I think we must use S-procedure for the proof above like in \cite{Xie1997}. And, is it really true for all cases? Maybe I should not insist on that result! Convexity is ruined... }}
\end{proof}
\begin{remark} \label{rem2}
LMI (\ref{eq:lmi1cons}) is also affine in $h$ for a fixed $\sigma$. So, the largest delay value $h_{\rm m}$, i.e. the {\rm{allowable delay bound}}, can be computed by solving the quasi-convex optimization problem,
\begin{equation}
\begin{array}{cc}
    \sup\limits_{\nu_l\in V_M} 					& \quad h  \\
    \rm{subject \; to}      & \quad P_l\succ 0, \;S_l\succ 0, \;R_l\succ 0, \; h>0, \; (\ref{eq:lmi1cons}),\\
		\end{array} 
\label{eq:maxh}
\end{equation}
using an iterative line search, in the context of Theorem~\ref{th:delay}.%\ba{Give a remark on calculation of the delay bound with infimum of $h$ w.r.t. vertices of $\lambda_j$ considering the quasi convexity.. See \cite{Li1997}.. }
\end{remark}
%
%\begin{remark}
%We can not have disconnected regions as in the delay-free case, since we are looking for a stable region for $\tau\in[0, h]$ with a simple LKF. This comes from the root continuity property. But, maybe we can for a sufficiently small delay bound. \label{rem3}
%\end{remark}
%%
%\begin{remark}
%Maybe, we can give another remark on common LKF approach, mentioning also the switching topologies, note however its conservativeness. \label{rem4}
%\end{remark}
%
%\begin{remark}
% MAYBE WE CAN PUT A REMARK ON VERTICES AND THE NONCONVEX DISTRIBUTION OF COMPLEX EIGENVALUES...
%\label{eq:}
%\end{remark}

Notice that in Theorem \ref{th:delay}, we examine the LMI feasibility point-wise with respect to the vertices of a convex-hull containing the graph matrix eigenvalues and the regarding distinct LMI matrix variables, i.e. $P_l, S_l, R_l$. Actually, this {\emph{multiple}} LKF approach is proposed to possibly reduce conservatism of the delay bound condition for guaranteed synchronization. Another approach is to consider a {\emph{common (parameter-independent)}} LKF with $P, S, R \succ 0$ such that LMI (\ref{eq:lmi1cons}) is feasible for all vertex points $\nu_l\in V_M$ in complex plane. However, this approach yields inherently more conservative results for the allowable delay bound. 
We will nevertheless utilize a single LKF, but not for the delay bound calculation; Section \ref{sec:control1} brings a control design based on a common LKF due to its advantages for a direct design with respect to vertices of a convex set in the complex plane.
% So, instead of the delay bound analysis, we will utilize the common LKF approach for the control design in Section \ref{sec:control} due to its advantages for a direct design with respect to vertices in the complex domain. 

%%%%%%%%%%%%%%%%%%%%%%%%%%%%%%%%%%%
\section{Delay-dependent Synchronizing Region Estimate} \label{sec:sr}%  Parameter-dependent and Common LKF Approach
In this section, we provide a solution to Problem \ref{prob2}. Namely, we present an algorithm to find DSR, $S_c(h)\subseteq\cc$ in (\ref{eq:sr}), where the MAS synchronization is guaranteed for a given delay bound $h$ and state-feedback controller $K$, i.e. the LMI (\ref{eq:lmi1cons}) in Theorem \ref{th:main} is feasible $\forall\sigma\in S_c(h)$. Recall that (\ref{eq:lmi1cons}) can be cast in an affine form with respect to the real and imaginary parts of $\sigma$ as in (\ref{eq:lmi1rep}). Thus, we can easily provide a {\emph{gridding approach}} as a straightforward relaxation for parameter-dependent LMIs, in the light of Theorems \ref{th:robust} and \ref{th:delay}, regarding the multi-quasi-convexity. Namely, we propose Algorithm~1 to find the points $(\gamma_k,\beta_k)\in\rr_{+}^2$, $\gamma:=\Re(\sigma)$ and $\beta:=\Im(\sigma)$ for given $h$ and $K$, where (\ref{eq:lmi1cons}) is feasible. \vspace{12pt}

\begin{tabular}{l p{10.5cm}}
\hline
{\textbf{Algorithm 1:}} & An Estimate of the DSR via Vertex Detection \\
\hline
{\textit{Initialization}}:
& Select a resolution value $\Delta\in\rr$ and set $\beta_0=0$. Find $\gamma_{min} \geq 0$ such that LMI (\ref{eq:lmi1cons}) is feasible with some $P\succ 0, \;S\succ 0, \;R\succ 0$. Set $k=0$, $\sigma = \gamma_{min}$, and save $\sigma_0 = \gamma_{min}$.
\end{tabular}
%\vspace{-3pt}

%\begin{algorithm}[t!]
%\caption{An Estimate of the DSR via Vertex Detection}\label{alg1}
%% \vspace{6pt}
%{\emph{Initialization}}: Select a resolution value $\Delta\in\rr$ and set $\beta_0=0$. Find $\gamma_{min} \geq 0$ such that LMI (\ref{eq:lmi1cons}) is feasible with some $P\succ 0, \;S\succ 0, \;R\succ 0$. Set $k=0$, $\sigma = \gamma_{min}$, and save $\sigma_0 = \gamma_{min}$. \\ % \vspace{12pt}

\begin{tabular}{ll}
 & {\emph{LOOP 1 : Increment of imaginary part}} \\
1: & {\textbf{while}} (\ref{eq:lmi1cons}) is feasible {\textbf{do}} \\% with $P\succ 0, \;S\succ 0, \;R\succ 0$\\
2: & \quad\quad $\sigma \leftarrow \sigma + i \Delta$. \\
3: & \quad\quad {\textbf{if}} (\ref{eq:lmi1cons}) is feasible go to step 2. \\
4: & \quad\quad {\textbf{else}} $k\leftarrow (k+1)$ {\textbf{then}} $\sigma_k = \sigma - i\Delta$ {\textbf{then}} $\sigma \leftarrow (\sigma + \Delta - i \Delta)$. \\
5: & \quad\quad\quad\quad {\textbf{if}} (\ref{eq:lmi1cons}) is feasible go to step 2. \\
6: & \quad\quad\quad\quad {\textbf{else}} $\sigma \leftarrow (\sigma - \Delta)$\\
7: & {\textbf{end while}} \\
 & {\emph{LOOP 2 : Decrement of real part}} \\
8: & {\textbf{while}} (\ref{eq:lmi1cons}) is feasible $\wedge \;\Im(\sigma)\geq 0 \;$  {\textbf{do}} \\% with $P\succ 0, \;S\succ 0, \;R\succ 0$\\
9: & \quad\quad $\sigma \leftarrow \sigma + \Delta$. \\
10: & \quad\quad {\textbf{if}} (\ref{eq:lmi1cons}) is feasible go to step 9. \\
11: & \quad\quad {\textbf{else}} $k\leftarrow (k+1)$ {\textbf{then}} $\sigma_k = \sigma - \Delta$ {\textbf{then}} $\sigma \leftarrow (\sigma - \Delta - i \Delta)$. \\
12: & \quad\quad\quad\quad {\textbf{if}} (\ref{eq:lmi1cons}) is feasible go to step 9. \\
%13: & \quad\quad\quad\quad {\textbf{else}} $\sigma \leftarrow (\sigma - \Delta)$\\
13: & {\textbf{end while}}
%14: & $\gamma_k:=\Re(\sigma_k)$, $\beta_k:=\Im(\sigma_k)$ {\textbf{then}} $\Lambda_M = {\rm{co}}\{ (\gamma_k,\beta_k), (\gamma_k,-\beta_k)\}$, $k=0,1, \ldots, M$ \\
\end{tabular} \vspace{12pt}

Notice that only the feasible vertex points in the $(\gamma,\beta)$ space are sought in the algorithm without checking the LMI feasibility for all points in a given range, making the algorithm more efficient. For this purpose, the first loop finds the maximum feasible $\beta$ for each increased value of $\gamma$ with the set resolution starting from a feasible point $(\gamma_{min},0)$. Then, after reaching the largest $\beta$ value for some $\gamma$, the first loop ends and then the maximum feasible $\gamma$ points for each decreased value of $\beta$ is found in the second loop. Note that we take advantage of having just two parameters on which the LMI depends affinely. So, the relaxation by gridding becomes quite easy and efficient \citep{Briat2015book}, reducing conservatism.
\begin{remark} \label{rem3}
Both loops in the algorithm find the vertices $\sigma_k$ where $\gamma_k \geq 0, \; \beta_k \geq 0$. Then, with respect to the hermitian property of the matrix in LMI (\ref{eq:lmi1cons}), the feasibility of the LMI also for $\sigma^{*}_k$ is trivial.
\end{remark}
% Notice also that

The region determined by the vertex points $\sigma_k$ found by the proposed algorithm may not be convex, especially due to limited resolution. However, with the result of Theorem \ref{th:delay} and the symmetry property mentioned in Remark \ref{rem3}, it is straightforward to conclude the following corollary to determine a convex estimate of the DSR, i.e.~a convex subset of $S_c(h)$ in (\ref{eq:sr}) for a given $h$.

%With this property and Theorem \ref{th:delay}, one can conclude that for any point in $\Lambda_M$ found in line-14 in Algorithm 1, (\ref{eq:lmi1cons}) is feasible, and thus the MAS synchronization is guaranteed. 

\begin{corollary} [An estimate of DSR] \label{cor1}
The MAS (\ref{eq:sadyn}--\ref{eq:locerr}) for a given arbitrary $K$ and $\forall\tau \in [0,h]$ synchronizes asymptotically if all the graph matrix eigenvalues satisfy $\lambda_j \in \Delta_M = {\mathbf{co}}\{\sigma_k\}$, $k=0, \ldots, M$, where $\sigma_k$ are obtained by Algorithm~1. Furthermore, such $\Delta_M$ is a convex estimate of the DSR (\ref{eq:sr}), i.e. $\Delta_M \subseteq S_c(h)$.
\end{corollary}

It is useful to determine the DSR in a less conservative way as in the approach proposed above, especially when the knowledge on the communication topology is restricted and uncertain. Also, one would like to know this region possibly for an appropriate design of the communication topology, i.e.~the design of the adjacency matrix $E$, with respect to the unknown delay by guaranteeing that each $\lambda_j$ resides in the found DSR ($\Delta_M$) by Corollary \ref{cor1}. %, where the synchronization is guaranteed. %and given the closed-loop dynamics ($A, B, K$) of the agents

%%%%%%%%%%%%%%%%%%%%%%%%%%%%%%%%%%%
\section{Robust Cooperative Synchronization Control Design} \label{sec:control}

In this section, we propose two different distributed control design approaches for the synchronization of MAS (\ref{eq:sadyn}--\ref{eq:locerr}), i.e.~two solutions to Problem \ref{prob3}.  
The first one considers a pre-defined convex-hull that contains all pinned Laplacian eigenvalues $\lambda_j$. It is based on a common LKF for the LMIs with the vertex matrices in the complex plane, which results in a direct design of the state-feedback controller. 
%which is a newly introduced technique among the literature of synchronizing region approach for MAS to the best of the authors' knowledge. This approach  
The second is based on the proposed multiple LKF approach in Sections \ref{sec:adb} and \ref{sec:sr}, together with the feedback gain scaling similarly as in the classical synchronizing region literature.

\subsection{Common LKF approach for state-feedback controller design} \label{sec:control1}
We consider the state-feedback control gain $K$ design for (\ref{eq:crfde}) to synchronize the MAS, for all delay values $0 \leq \tau \leq h$ and for a region of $\sigma\in\cc$ that contains all $\lambda_j$'s. Note that LMI (\ref{eq:lmi1cons}) with $A_d=-BK$ is nonlinear if $K$ is considered variable, and it is not easy to linearize it via variable transformations without restrictive assumptions on $P, \;S$ and $R$; see \cite{Fridman2014book}. Thus, we utilize the {\emph{descriptor method}} proposed in \cite{Fridman2002} following the lines of \cite{Fridman2014book} to provide an efficient design approach, which may reduce conservatism.

Let us recall the LKF (\ref{eq:lkfcons}) and add the descriptor equation given by % with $P\succ 0, S\succ 0, R\succ 0$
\begin{equation}
\begin{split}
0 & = \left[x^T(t)Y^T + \dot{x}^T(t) \epsilon Y^T \right] \left[A x(t) + \sigma A_d x(t-h) - \dot{x}(t) \right] \\
 & + \left[x^T(t) A^T + \sigma x^T(t-h) A_d^T - \dot{x}(t)\right] \left[Y x(t) + \epsilon Y \dot{x}(t) \right],
\end{split}
\label{eq:desc}
\end{equation} 
where $Y\in\rr^{n\times n}$ is a full-rank symmetric matrix and $\epsilon > 0$, to the right hand side of inequality (\ref{eq:dlkfineqcons}). Denoting $\bar{\eta}(t) = [x(t) \; x(t-h) \; \dot{x}(t)]$, we arrive at
\begin{equation}
\frac{dV}{dt} \leq \bar\eta^H(t)
 \begin{bmatrix}
    A^T Y + Y A + S - R  & \sigma Y A_d + R  & -Y + P +  \epsilon A^T Y\\
    \star             					& -S-R   & \epsilon \sigma^{*} A_d^T Y \\
		\star             					& \star   & -\epsilon 2 Y + h^2 R \\
\end{bmatrix}
\bar\eta (t),
\label{eq:desclyader}
\end{equation}
where $A_d = - BK$. Denote $\bar Y = Y^{-1}$, $\bar X = K Y^{-1}$, $[\bar{P} \;\bar{S} \; \bar{R}] = \bar{Y}^T [ P \; S \; R] \bar{Y}$, and multiply the matrix in (\ref{eq:desclyader}) by diag $\{\bar{Y},\bar{Y},\bar{Y}\}$ from the right and its transpose from the left. Consequently, we get the LMI
\begin{equation}
 \begin{bmatrix}
    \bar{Y} A^T + A \bar{Y} + \bar S - \bar R  & -\sigma B\bar{X} + \bar R  & \bar P -\bar Y^T +  \epsilon \bar Y A^T \\
    \star             					& -(\bar S + \bar R)   & -\sigma^{*} \epsilon \bar{X}^T B^T \\
		\star             					& \star   & -\epsilon 2 \bar Y + h^2 \bar R \\
\end{bmatrix} \prec 0,
\label{eq:desclmi}
\end{equation}
with variables $\bar P \succ 0, \bar S \succ 0,\bar R \succ 0$, any invertible symmetric $\bar Y \in \rr^{n\times n}$, and any $\bar X \in \rr^{m\times n}$, for a given tuning scalar parameter $\epsilon>0$, which gives a sufficient condition for the negative definiteness of the time-derivative of the LKF for complex RFDE (\ref{eq:crfde}). Notice that LMI (\ref{eq:desclmi}) is affine in $\sigma$ (and $\sigma^{*}$) such that it can be written similarly as (\ref{eq:lmi1rep}), which preserves multi-quasi-convexity. So, we conclude in light of Theorem \ref{th:robust} that if (\ref{eq:desclmi}) is feasible with $\bar P \succ 0, \bar S \succ 0,\bar R \succ 0$ for all vertices of a convex polygon generated by the real and imaginary parts of $\sigma$, then  it is also feasible for all points inside the determined convex polygon. This means there exists a common LKF (\ref{eq:lkfcons}) guaranteeing the stability of the complex RFDE (\ref{eq:crfde}) for all those points. With these results, we can give the following theorem for the cooperative tracking control design for MAS. 
%
%Let $\Lambda_N = \{\lambda_j\in\cc\}$, $j=1,\ldots,N$ be the set of all eigenvalues of L+G regarding a fixed graph topology of MAS with (\ref{eq:sadyn}-\ref{eq:locerr}). Define its convex hull as $\Delta:={\mathbf{co}}\{\Lambda_N\}$, and the corresponding set of vertices as $V_M:={\mathbf{vert}}\{\Delta\}$. 
%Given $h\geq 0$ and $A_d = -BK$, let there exists $P_l\succ 0, S_l\succ 0,R_l\succ 0$, such that LMI (\ref{eq:lmi1cons}) with $\sigma\leftarrow\lambda_l$, $\forall \lambda_l\in V_M$, $l=1,\ldots,M$, is feasible.
%
\begin{theorem} [Direct state-feedback design for synchronization]
\label{th:control}
Given $\epsilon>0$, $h>0$; let there exists common $n\times n$ matrices $\bar P\succ 0, \; \bar S\succ 0, \;\bar R\succ 0$, any symmetric full-rank $\bar{Y}\in \rr^{n\times n}$, and any $\bar X \in \rr^{m\times n}$ such that LMIs (\ref{eq:desclmi}) with $\sigma\leftarrow\nu_l$, $l = 1\ldots M$, are feasible, where $\nu_l \in V_M$ are the vertices of a convex hull and all the pinned Laplacian eigenvalues satisfy $\lambda_j \in {\mathbf{co}}\{V_M\}$. Then, the state-feedback gain $K = \bar{X}\bar{Y}^{-1}$ guarantees the synchronization of all agents (\ref{eq:sadyn}) to the leader node (\ref{eq:ldyn}) for all delays $\tau \in [0,h]$. 
\end{theorem}
%
%\begin{remark} \label{rem5}
%Note that Theorem \ref{th:control} is valid also for switching graph topologies in the case that .... since it provides a common LKF independent of time that guarantees the synchronization. 
%\end{remark} 
%
Note that tuning the parameter $\epsilon$ may yield less conservative results for delay bound \citep{Fridman2014book}. %, which will be illustrated in the case study example. 
However, even with thus reduced conservatism, for larger delays it may not be possible to find a common LKF as proposed in Theorem \ref{th:control}. 
Thus in the following subsection, we present an alternative design approach which offers a less conservative delay bound.
% Thus, in the following subsection we present the scaling of the Laplacian eigenvalues with a constant gain multiplying $K$ into the DSR found in Section \ref{sec:sr}, which offers a less conservative approach with the relaxation. %  by gridding
%
\subsection{Controller design via the scaling graph matrix eigenvalues}\label{sec:control2}
Here, we provide a distributed control design via scaling graph matrix eigenvalues into the relaxed DSR estimate proposed in Section \ref{sec:sr}. Consider the local state-feedback control law (\ref{eq:fb}) for each agent in the form 
\begin{equation}
u_i(t)=c \bar K e_i(t),
\label{eq:fbbar}
\end{equation}
where $c>0$ is the constant {\emph{coupling gain}}. Then, one can consider the complex RFDE (\ref{eq:crfde}) with $\sigma = c \lambda_j $ and $K = \bar{K}$ for the synchronization analysis and control. The strategy is to design a stabilizing controller $\bar{K}$ for a real value of $\sigma$, e.g. name as $\sigma_R$, and then to scale all the Laplacian eigenvalues with $c$ inside the estimated DSR. So, we propose the following procedure.

\medskip
\textbf{Procedure 1:} 
\begin{enumerate}[label=(\roman*)]
	\item Design $\bar{K} = \bar{X}\bar{Y}^{-1} $ via solving the LMI (\ref{eq:desclmi}) with $\sigma\leftarrow\sigma_R= [\rm{min}\;\Re(\lambda_j) + \rm{max}\;\Re(\lambda_j)]/2$.
	\item Find the DSR, i.e. $\Delta_M \subseteq S_c(h)$, by Corollary \ref{cor1} for the calculated $\bar{K}$ in (i). 
	\item Find a value of $c$ by increasing or decreasing starting from $c_0=1$ such that all $\sigma_j=c\lambda_j$ are scaled into $\Delta_M$ found in (ii).
	\end{enumerate}
%
%{\textbf{Procedure 1:}}
%\begin{enumerate}[(i)]
	%\item Design $\bar{K} = \bar{X}\bar{Y}^{-1} $ via solving the LMI (\ref{eq:desclmi}) with $\sigma\leftarrow\sigma_R= [\rm{min}\;\Re(\lambda_j) + \rm{max}\;\Re(\lambda_j)]/2$.
	%\item Find the DSR, i.e. $\Delta_M \subseteq S_c(h)$, by Corollary \ref{cor1} for the calculated $\bar{K}$ in (i). 
	%\item Find a value of $c$ by increasing or decreasing starting from $c_0=1$ such that all $\sigma_j=c\lambda_j$ are scaled into $\Delta_M$ found in (ii). \qed 
%\end{enumerate}
%
\begin{remark} \label{remex}
The reason behind the choice of $\sigma_R$ in step (i) above is to guarantee the synchronization for a center-like point of all graph matrix eigenvalues in the complex domain. This choice enables one to find a coupling gain $c$ (possibly $c=1$ in most cases) easier/faster that scales all $\lambda_j$'s into the $\Delta_M$ found in step (ii). 
\end{remark}

One can easily find the upper (lower) bound on $c$ with an increment (decrement) procedure in step (iii), using the {\emph{inpolygon}} command in a simple {\emph{while}} loop in MATLAB. It is worth to point out, for comparison that in \cite{Movric15jfi}, the coupling constant is pre-chosen as $c\geq \frac{1}{2 \rm{min} \; \Re (\lambda_j)}$, which is quite a restrictive lower bound. One may need to choose $c$ smaller in some cases, such as to avoid exceeding the limit of the control input amplitude. 

%, e.g. for a faster synchronization. Similarly, the lower bound of $c$ can be found via a decrement procedure. 
%We will illustrate the obtainment of lower and upper bounds of coupling constant in the case study below. 
%
\begin{remark}
In Procedure~1, we make use of multi-quasi-convex nature of the problem to find a single $K$ for multiple LKFs, i.e.~feasible LMIs.~This relaxes the conventional synchronizing region approach, which uses a common (single) Lyapunov function, LRF or LKF, see e.g. \cite{Li2010,Zhang2011,Hengster-Movric2014,Movric15aut,Movric15jfi}. This is found to considerably reduce the conservatism encountered in the conventional synchronizing regions when applied to agents with delays \citep{Movric15jfi}, c.f.~Section \ref{sec:ex}. 
%, or optimal control design for the delay-free MAS to reduce the conservatism in synchronizing region found with conventional approaches. 
\label{rem4}
\end{remark}
\begin{remark} \label{rem5}
The state-feedback gain $K$ in Theorem~\ref{th:control} can be scaled with a coupling gain $c>1$, such that all $c\lambda_j$ are kept inside the $\Delta_M$ found by Corollary~\ref{cor1}. This is because $K$ found by a common LKF approach may result in a larger DSR which can be estimated by the proposed gridding approach. However, common LKF approach may lead to more conservative results for allowable delay bound as compared to the multiple LKF approach. Thus, a design decision concerning allowable delay bound and graph topology uncertainties is needed with respect to the trade-off between the conservatism of the proposed methods. This fact is illustrated in Section \ref{sec:ex}.
\end{remark}

Finally, let us stress that both control approaches are robust with respect to the allowable delay and the network topology. Both methods guarantee the synchronization for any delay value $0\leq\tau\leq h$, and for any topology having the graph matrix eigenvalues inside a prescribed region. 

%%%%%%%%%%%%%%%%%%%%%%%%%%%%%%%%%%%
\section{Numerical Example} \label{sec:ex}
To demonstrate the (efficacy of) proposed methods, we  borrow an example from \cite{Movric15jfi}, where four agents having the LTI dynamics % in the sense of reducing conservatism
\begin{equation} \label{eq:compex} 
\dot{x}_i(t)=\begin{bmatrix}
    0  & 1  \\
    -1 & 0 \\
\end{bmatrix} x_i(t) + \begin{bmatrix} 0  \\ 1 \\ \end{bmatrix} u_i(t-\tau), \quad i=1,\ldots,4,
\end{equation}
are considered. The Laplacian matrix for the directed communication topology is described as $L = [2 \; -1 \; 0 \; -1 \; ; \; 0 \; 1 \; -1 \; 0 \; ; \; -1 \; 0 \; 1 \; 0 \; ; \; 0 \; -1 \; -1 \; 2]$. The pinning gain matrix is $G = \text{diag}(1,1,0,0)$.
The pinned Laplacian eigenvalues are $\lambda_1=0.368$, $\lambda_{2,3}=2.5 \pm i 0.866$, and $\lambda_4=2.618$. 
The local feedback gain matrix and the coupling constant for the control in (\ref{eq:fb}) is obtained in \cite{Movric15jfi} as $\bar K=[0.1 \; 0.6403]$ and $c=1.34$, respectively. Using those values for the control law (\ref{eq:fb}) with $K=c\bar K$, the allowable delay bound for synchronization is here found as $h_{\rm m} = 0.4190$ as described in Remark \ref{rem2} via the feasibility of the LMI (\ref{eq:lmi1cons}) for all $\lambda_j$'s. This is beacuse the sets $\Lambda_N$ and $V_M$ in Theorem \ref{th:delay} are the same for this  communication topology. Note that the delay bound found in \cite{Movric15jfi} by the LRF approach is $h_{\rm m}=0.1137$, while the `actual' delay bound found by the spectral domain analysis is $\tau_{\rm m} = 0.4445$ given also in \cite{Movric15jfi}. Thus, a significant improvement in reducing conservatism of the condition on the allowable delay bound is achieved. 
\begin{figure}%[!b]
\centering
\includegraphics[width=0.475\columnwidth]{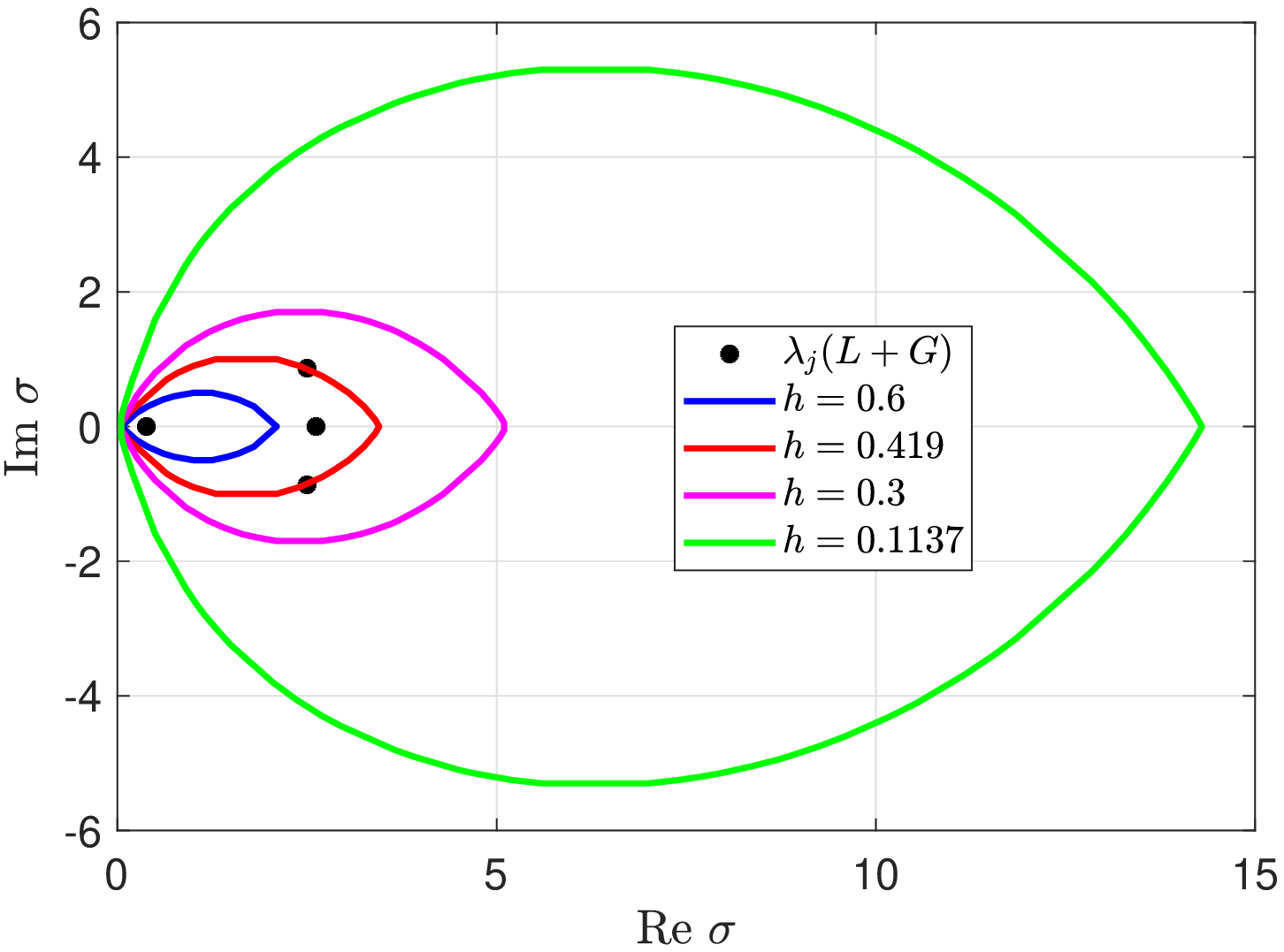}%
\includegraphics[width=0.475\columnwidth]{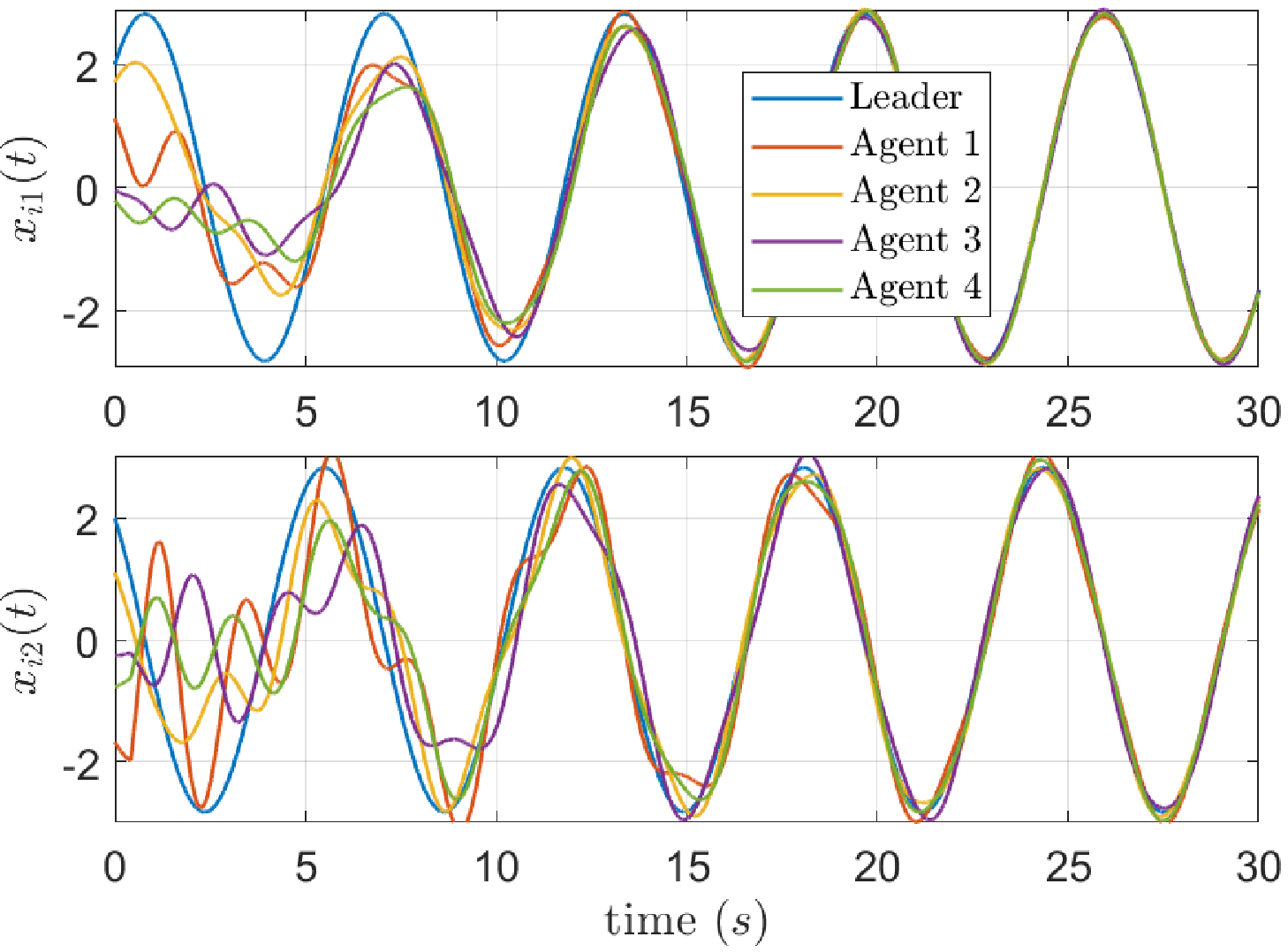}%
\vspace{-6pt}
\caption{Left: Synchronizing regions ($\Delta_M$) found by Corollary~\ref{cor1} for various values of the delay bound $h$ for the controller $K=1.34\;[0.1 \; 0.6403]$. Right: State variables of the leader and agents for $\tau = 0.419$ s.}%
\label{fig1}%
\end{figure}
%
%\begin{table}%[!b]
%\begin{center}
%\caption{Controllers designed by the proposed approaches in Section~\ref{sec:control} for the MAS synchronization.} \vspace{6pt}
%\label{tab1}\begin{tabular}{|p{2.28cm} || p{4.25cm} | p{4.25cm} |}
%\hline
 %Design method & $h=0.6 \quad (\epsilon=0.1) $ & $h=0.9 \quad (\epsilon=0.1) $ \\
%\hline\hline
%Theorem~\ref{th:control} & $K = [-0.0173  \;  0.2531]$ & no feasible solution \\
%\hline
%Procedure 1 &  $\bar K = [-0.005  \;  0.2883], \; c=1 $ & $\bar K = [-0.011 \; 0.1446], \; c = 1 $\\ 
 %& $ c_{min}=0.131, \; c_{max}=1.873 $ & $ c_{min}=0.131, \; c_{max}=1.024 $ \\
%\hline
%\end{tabular}
%\end{center}
%\end{table} 
%
\begin{table}%[!b]
\begin{center}
\caption{Controllers designed by the proposed approaches in Section~\ref{sec:control} for the MAS synchronization.} \vspace{6pt}
\label{tab1}\begin{tabular}{|p{2.5cm} || p{5.1cm} | p{5.1cm} |}
\hline
 Design method & $h=0.6 \quad (\epsilon=0.1) $ & $h=0.9 \quad (\epsilon=0.1) $ \\
\hline\hline
Theorem~\ref{th:control} & $K = [-0.0173  \;  0.2531]$ & no feasible solution \\
\hline
Procedure~1 &  $\bar K = [-0.005  \;  0.2883], \quad c=1 $ & $\bar K = [-0.011 \; 0.1446], \quad c = 1 $\\ 
 & $ c_{min}=0.131, \quad c_{max}=1.873 $ & $ c_{min}=0.131, \quad c_{max}=1.024 $ \\
\hline
\end{tabular}
\end{center}
\end{table} 

\begin{figure}[!t]
\centering
\includegraphics[width=0.47\columnwidth]{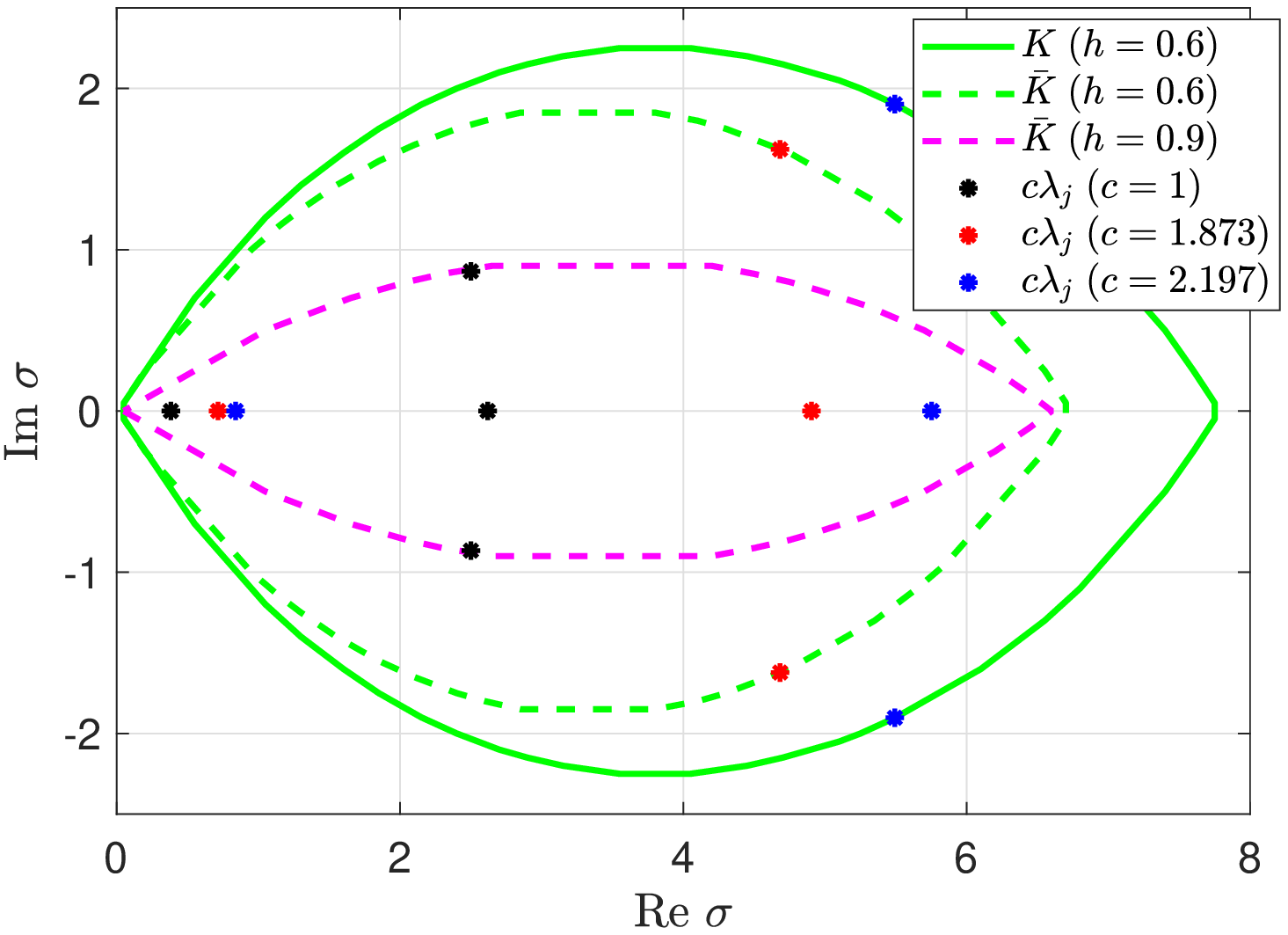}% synch_states_h06_Kclf
\includegraphics[width=0.47\columnwidth]{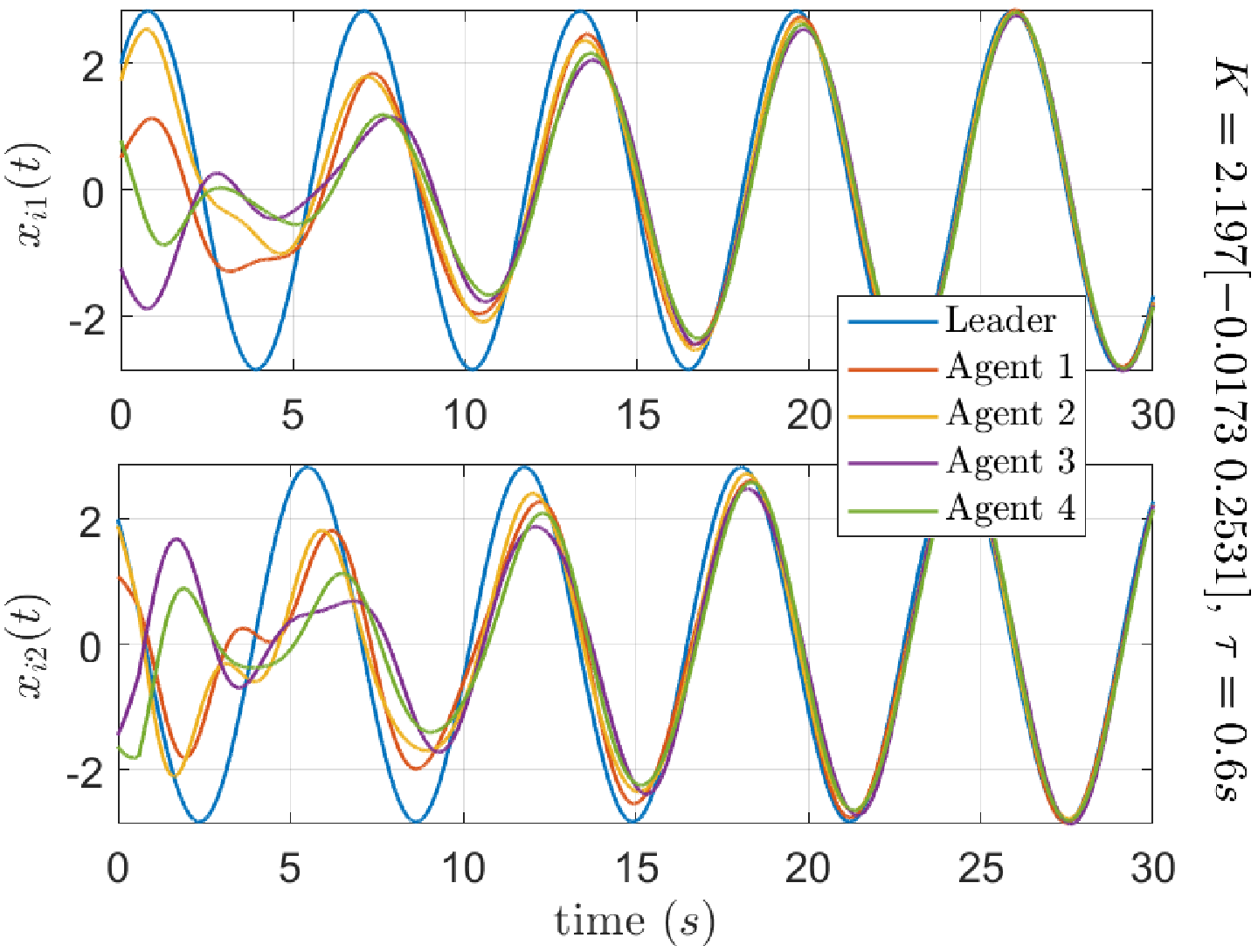}
\includegraphics[width=0.47\columnwidth]{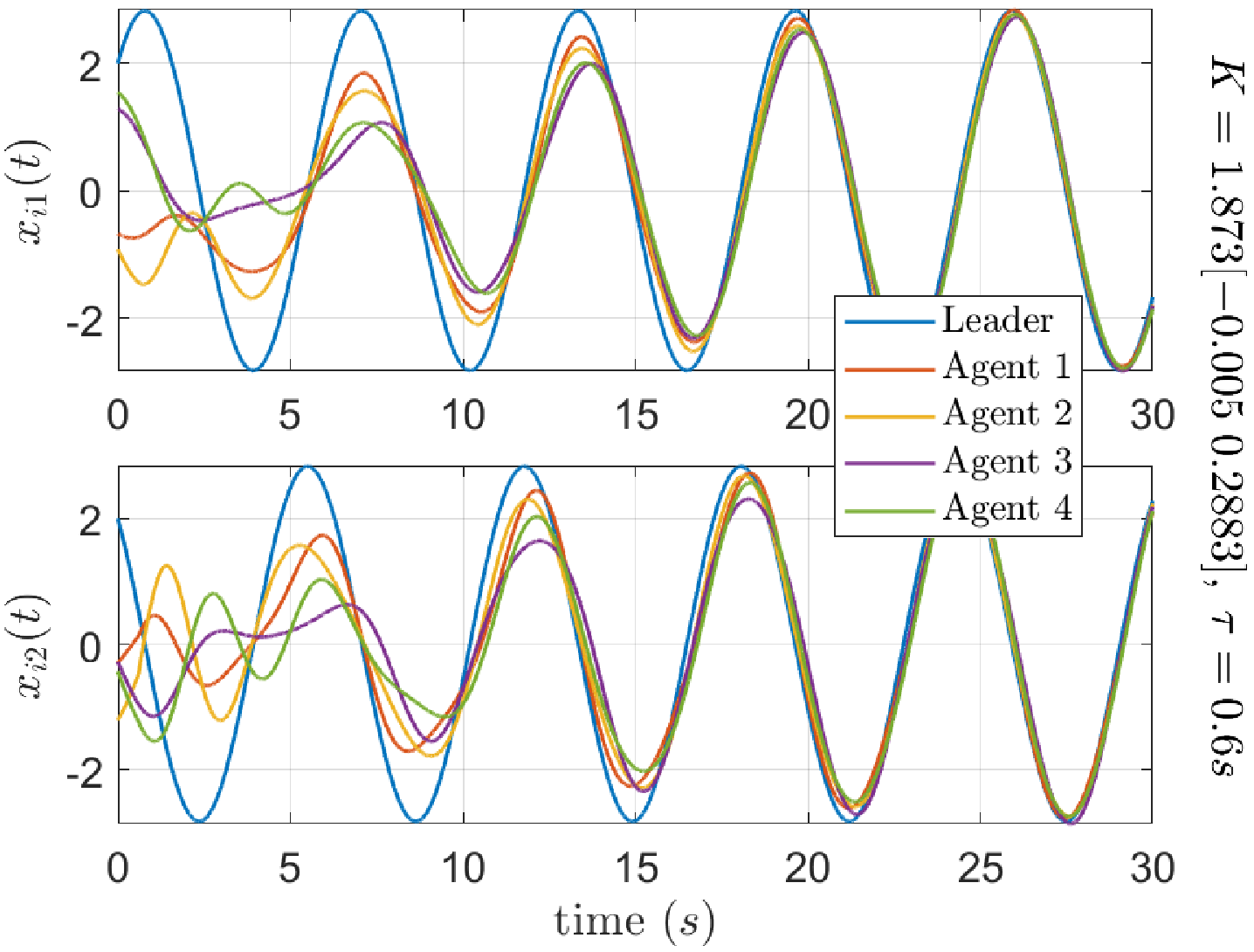}
\includegraphics[width=0.47\columnwidth]{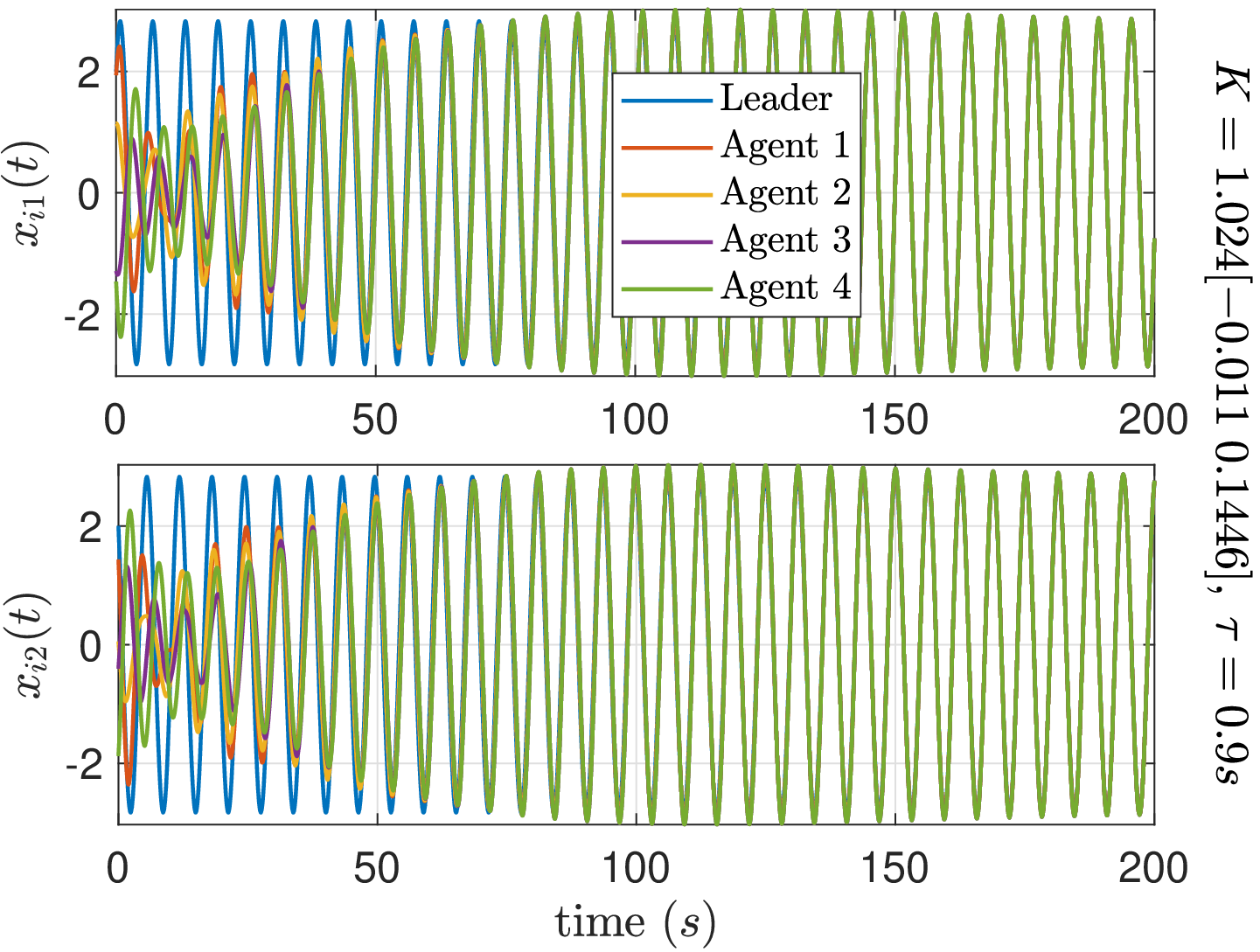}
\vspace{-9pt}
\caption{Upper-left: Synchronizing regions ($\Delta_M$) found by Corollary~\ref{cor1} for the control designs via Theorem~\ref{th:control} and Procedure~1 as in Table~\ref{tab1} for the upper delay bounds $h=0.6$ and $h=0.9$. Others: State variables of the leader and agents for $\tau = h$, and controllers amplified with $c_{max}$ given Table~\ref{tab1}.}%
\label{fig2}%
\end{figure}

We show in Fig.~\ref{fig1}--left the DSRs found by Corollary~\ref{cor1} with the resolution $\Delta = 0.05$ for the proposed algorithm when $K = 1.34[0.1 \; 0.6403]$ is chosen. Observe that $\lambda_{2,3}$ lie on the boundary of the region whereas $\lambda_{1,4}$ lies inside the region, found for the delay $h=0.419$, which is consistent with the presented result above found by Theorem \ref{th:delay}. 
%\ba{For a comparison, we illustrate the DSRs for $h=0.1137$ obtained by Corollary 2 (with a nonlinear matrix inequality) in \cite{Movric15jfi} and by the proposed Corollary \ref{cor1}}. 
For comparison, note that the DSR is determined by nonlinear matrix inequalities in \cite{Movric15jfi}, and is not easy to determine in its entirety. Note also that the delay bound $h=0.1137$ in \cite{Movric15jfi} is found for the condition $|\lambda_{max}|\leq 2.6457$. Thus, the conservatism of the DSR result is here reduced remarkably. Fig~\ref{fig1}--right also demonstrates the synchronization of the both state variables of agents to the states of the leader when $\tau = 0.419$ s in (\ref{eq:locerr}), where the initial state of the leader is $x_0(0) = [2 \; 2]^T$, and each initial state values of the agents are random values in the interval $[-2,2]$.

Let us also design the controllers proposed in Section~\ref{sec:control} for the delay bounds $h=0.6$ and $h=0.9$. The state-feedback controllers and coupling gains are obtained for both proposed methods as in Table~\ref{tab1}.  
The DSRs found by Corollary~\ref{cor1} and $c\lambda_j$ values with $c_{max}$ are illustrated in Fig~\ref{fig2}-upperleft. It is seen that all regions contain $\forall\lambda_j$ and $c_{max}\lambda_{2,3}$ lie on the boundary of $\Delta_M$. The largest DSR is found for $K$ by means of the common LKF approach, i.e. Theorem~\ref{th:control}, as mentioned in Remark \ref{rem5}. Correspondingly, $K$ can also be multiplied with any coupling gain $1 \leq c\leq 2.197$, where the synchronization is still guaranteed for $h=0.6$. The other sub-figures of Fig.~\ref{fig2} show that the synchronization of state variables of all the agents to the leader states is achieved for all controllers multiplied by maximum coupling constants.
% {\em{It would be good to demonstrate the results for a large scale MAS, especially to show that it is enough to consider just the vertex eigenvalues, or a convex-hull covering the eigenvalues..}}

All the LMI problems above are solved by standard SeDuMi solver with YALMIP toolbox. 
%
%%%%%%%%%%%%%%%%%%%%%%%%%%%
\section{Concluding remarks}\label{sec:conc}

Synchronization of MAS with constant uniform communication and control delays is investigated. Agents are described by general LTI dynamics, and assumed to interact on a directed graph. A unified LMI approach is proposed for determining the allowable delay bound, delay-dependent synchronizing region for any given state-feedback controller, and the design of distributed cooperative state-feedback control.~The proposed method is found to be less conservative both for the delay bound and synchronizing region estimates, as compared to other results existing in the literature. This is due to the LKF approach and the relaxation based on the revealed quasi-convexity property. For control design, a multiple LKF and a common LKF approach are proposed. While the multiple LKF results in a less conservative delay bound, the common LKF approach enables a more direct design and provides better robustness to graph uncertainties. For both approaches, it is sufficient to know only the upper bound on the delay and the approximate region in the complex plane where the Laplacian eigenvalues reside. Thus, those are both robust to delay and graph uncertainties. 

The proposed approach can also be adapted straightforwardly to the ARE-based design of state-feedback and output feedback controllers in \cite{Li2010,Zhang2011}. Also, it has high potential to reduce conservatism in various MAS synchronization control designs such as optimal and $H_\infty$ control studied in \cite{Hengster-Movric2014,Movric15aut}. Thus, we believe this paper will have some impact on the synchronization problem of MASs without delays as well. Another potential research area is the extension of the methodology to non-uniform constant or time-varying delays. Subsequent research on the subject will follow these directions.
\section*{Acknowledgments}
This work is supported by the European Regional Development Fund under the project Robotics for Industry 4.0 (reg. no. CZ.02.1.01/0.0/0.0/15\_003/0000470). The work of the second author is supported by The Czech Science Foundation project No. GA19-18424S.

% \section*{Bibliography}
% \bibliographystyle{abbrv}
\bibliographystyle{tfcad}
\bibliography{references_v2}

\end{document}